\documentclass[lettersize,journal]{IEEEtran}
\usepackage{amsmath,amsfonts,amssymb,mathrsfs, dsfont}
\usepackage{algorithmic}
\usepackage{algorithm}
\usepackage{array}
\usepackage[caption=false,font=normalsize,labelfont=sf,textfont=sf]{subfig}
\usepackage{textcomp}
\usepackage{stfloats}
\usepackage{url}
\usepackage{verbatim}
\usepackage{graphicx}
\usepackage{cite}
\usepackage{tikz}
\usepackage{stackengine}
\usepackage{mathtools}
\usepackage{amsmath}
\usepackage{graphicx}
\usepackage[colorlinks=true, allcolors=blue]{hyperref}
\usepackage{float}
\usepackage{amsmath}
\usepackage{amssymb}
\usepackage{bm}
\usepackage{graphicx}
\usepackage{float}
\usepackage{tabularx}
\usepackage{adjustbox}
\usepackage{changepage}
\usepackage{bbm}
\usepackage{amsthm}

\usepackage[switch]{lineno}
\newtheorem{theorem}{Theorem}

\newtheorem{definition}{Definition}
\newtheorem{proposition}[theorem]{Proposition}
\newtheorem{remark}{Remark}
\newtheorem{assumption}{Assumption}

\newtheorem{example}{Example}

\newcommand{\cm}[1]{{{{#1}}}}

\newcommand{\cmt}[1]{{{#1}}}

\usepackage{eqnarray}
\newcommand{\ra}[1]{\renewcommand{\arraystretch}{#1}}

\usepackage{xfrac}
\newcommand{\R}{\mathbb{R}}

\newcommand{\mj}{|\mathcal{N}|}
\newcommand{\me}{|\mathcal{E}|}
\newcommand{\mc}{m_\mathrm{c}}
\newcommand{\mep}{|\mathcal{E}^+|}
\newcommand{\mr}{m_r}
\newcommand{\mf}{m_f}
\newcommand{\xj}{x^{\mathrm{j}}}
\newcommand{\xc}{x^{\mathrm{c}}}

\usepackage{tabularx}
\usepackage{booktabs}
\usepackage{multirow}
\usepackage{courier}
\usepackage{float}
\usepackage{makecell}
\setcellgapes{5pt}
\begin{document}

\title{Economic Nonlinear Model Predictive Control of Prosumer District Heating Networks:\\ The Extended Version}

\author{Max Sibeijn, Saeed Ahmed, Mohammad Khosravi, and Tam\'as Keviczky

\thanks{
This work was funded by the Local Inclusive Future
Energy (LIFE) City Project (MOOI32019), funded by the
Ministry of Economic Affairs and Climate and by the Ministry
of the Interior and Kingdom Relations of the Netherlands.}%
\thanks{Max Sibeijn, Mohammad Khosravi, and Tam\'as Keviczky are with the Delft Center for Systems and Control,
        Delft University of Technology,  2628 CN Delft, The Netherlands (e-mail: 
        { $\{$m.w.sibeijn, mohammad.khosravi, t.keviczky$\}$@tudelft.nl}).}%
\thanks{Saeed Ahmed is with the Jan C. Willems Center for Systems and Control, ENTEG, Faculty of Science and Engineering,
       University of Groningen, 9747 AG Groningen, The Netherlands (e-mail: 
   {s.ahmed@rug.nl}).}%

}



\maketitle

\begin{abstract}
In this paper, we propose an economic nonlinear model predictive control (MPC) algorithm for district heating networks (DHNs). The proposed method features prosumers, multiple producers, and storage systems, which are essential components of 4th generation DHNs. These  networks are characterized by their ability to \cm{optimize their operations}, aiming to reduce supply temperatures, accommodate distributed heat sources, and leverage the flexibility provided by thermal inertia and storage—all crucial for achieving a fossil-fuel-free energy supply. 
Developing a smart energy management system to accomplish these goals requires detailed models of highly complex nonlinear systems and computational algorithms able to handle large-scale optimization problems. To address this, we introduce a graph-based optimization-oriented model that efficiently integrates distributed producers, prosumers, storage buffers, and bidirectional pipe flows, such that it can be implemented in a real-time MPC setting. Furthermore, we conduct several numerical experiments to evaluate the performance of the proposed algorithms in closed-loop. \cm{Our findings demonstrate that the MPC methods achieved up to 9\% cost improvement over traditional rule-based controllers while better maintaining system constraints. }
\end{abstract}


\section{Introduction}
\IEEEPARstart{T}{he} energy transition requires a major shift from fossil fuel-based generation to renewable energy sources. In particular, the electrification of heat production is expected to grow enormously; see, e.g., \cite{ISO2024}. Given that the thermal energy sector contributes roughly 50\% to the total \cm{final} energy consumption \cm{within the EU}~\cite{Paardekooper2018}, the need for sustainable heating becomes ever more substantial. Simultaneously, a global increase in the burden on power grids is evident, driven by rising electricity demand (e.g., heat pumps, electric vehicles) and fluctuating supply from renewables (e.g., wind, solar).

In order to address the higher demand and capacity limitations, a possible solution is to expand the existing infrastructure of power grids. However, expanding power grids is typically costly, time-consuming, and held back by a lack of appropriate regulatory frameworks~\cite{Battaglini2012}. An alternative solution is offered by the energy flexibility of district heating networks~(DHNs). DHNs are networks of pipelines that transport heated water by circulating it around \cm{a} district or city. Heat is either taken from or added to the network through heat exchangers located at consumers and producers, respectively. An example of a DHN 
is shown in Figure~\ref{fig:DHN_Scheme}.

District heating offers numerous advantages. Economically, it is cost-effective to implement large installations that require less capacity due to simultaneous use. Additionally, DHNs possess significant thermal inertia, allowing them to efficiently utilize free or inexpensive heat from industrial processes or underground geothermal sources. This characteristic provides DHNs with a synergistic element, often referred to as the economy of scope \cite{Werner2013}. Moreover, DHNs can create and exploit flexibility on daily operational timescales by responding to market conditions, which can, in turn, alleviate \cm{strain} on the electricity grid. Nevertheless, DHNs are large-scale systems with many distributed controllable assets. As a result, advanced control strategies are necessary for the effective operational management of these networks. 

\IEEEpubidadjcol

\begin{figure}[thbp]
    \centering
    \includegraphics[width=.98\linewidth, trim={6.5cm 0.8cm 8.5cm 2.2cm}, clip]{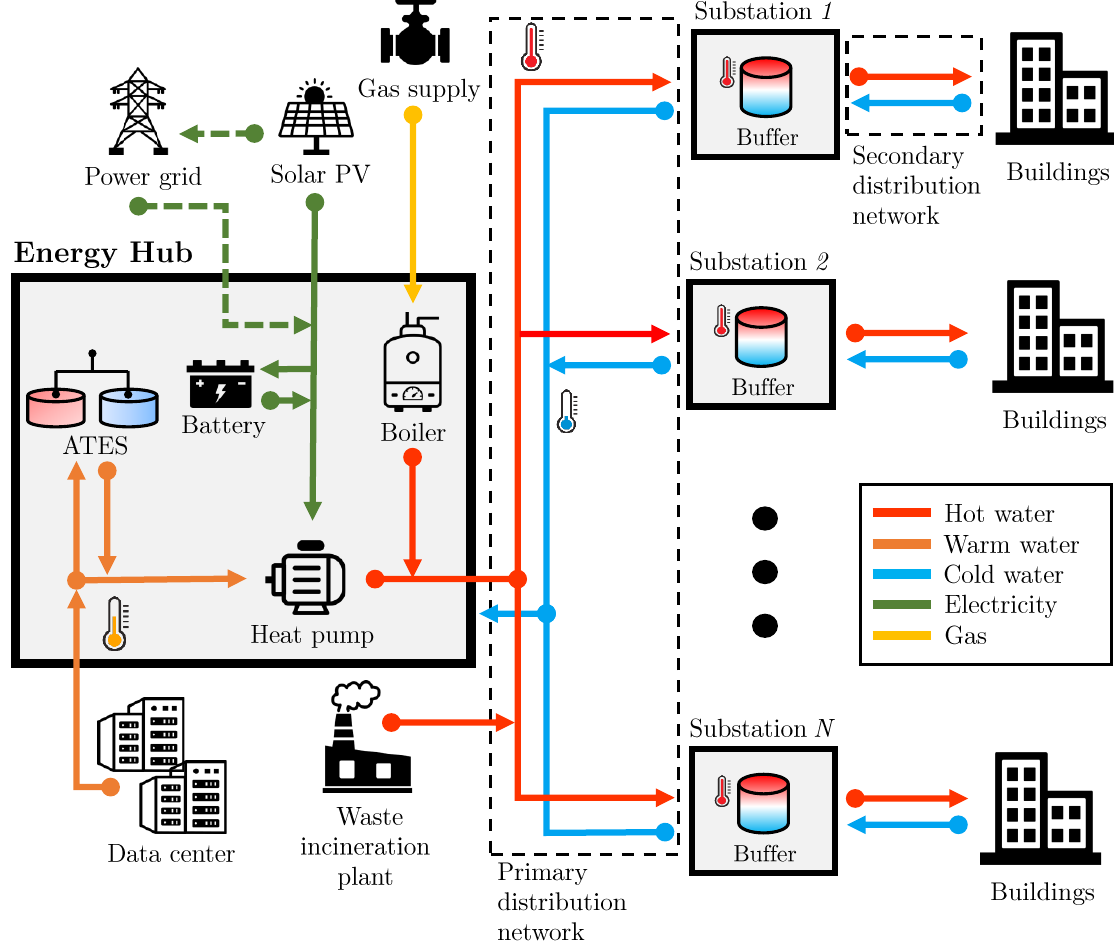}
    \caption{A schematic diagram of a district heating network with multiple substations (consumers) and production points.}
    \label{fig:DHN_Scheme}
\end{figure}

In the past, operational management of DHNs used rule-based controllers that set supply temperatures based on the outside temperatures~\cite{Novitsky2020}, while mass flows were controlled locally at substations using simple tracking controllers. 
For the new generation of district heating networks, referred to as \textit{4th generation district heating}~\cite{Lund2014}, these control methods are insufficient. Several studies~\cite{Novitsky2020,Lund2014,Vandermeulen2018}, emphasize the need for \cm{advanced} control to achieve necessary goals such as supply temperature minimization, creation of flexibility, distribution of heat generation, use of renewable energy sources, and fair sharing of heat. 
Closed-loop optimization-based control methods, such as model predictive control~(MPC), are at the intersection of stabilizing control and operational optimization~\cite{Schwenzer2021}. MPC combines the classical benefits of stabilizing control—such as robustness to demand prediction errors and model mismatch—with the advantages of optimization techniques. Therefore, MPC presents itself as a suitable and effective strategy for managing DHNs~\cite{Vandermeulen2018}.


\subsection{State of the Art}

Classical control methods for DHNs relied on controlling the temperature and the differential pressure at central supply units. For details on these methods, we refer to \cite{Werner2013, Vandermeulen2018}, and the references therein.
Recently, graph-based modeling techniques for DHNs have been a useful tool for developing various stabilizing controllers for both hydraulic and thermal management of DHNs. In this context, the pressure and flow regulation problem was addressed in \cite{Persis2011, Strehle2024}, temperature regulation was studied in \cite{Simonsson2024b, Ahmed2023}, and the stabilizing control of both storage and temperature in a DHN was studied in \cite{Scholten2016, Machado2022b, Machado2022}. Moreover, the works \cite{Machado2022b, Machado2022} introduced a multi-producer graph model integrating the dynamic evolution of storage volumes.

Operational optimization of DHNs has been a topic of interest for several decades, with earlier works dating back to the 1990s~\cite{Ben1995, Zhao1998}. These studies highlighted the complexity of managing load distribution while simultaneously minimizing supply temperatures, a challenge stemming from the nonlinear nature of thermal transients and their dependence on flow rates.
Recent studies, such as \cite{Krug2020}, have addressed thermal transients and variable flow rates through the development of open-loop optimization-based controllers using graph-theoretic models based on partial differential equations (PDEs) governing the one-dimensional pipe dynamics. This work also employed a complementary constrained formulation, originating from literature on gas transportation networks~\cite{Hante2019}, to manage switching flow directions. 
However, closed-loop implementation and important 4th generation DHN features, such as {multiple producers or storage}, were not studied in \cite{Krug2020}.

Recent research has addressed the computational challenges of DHN optimization. Adaptive optimization methods have improved model resolution at critical grid points \cite{Danschel2023}, while model-reduction techniques using Galerkin projections have significantly enhanced numerical solver efficiency \cite{c0}. Graph-theoretic DHN models have demonstrated promising results, achieving good accuracy even with simplified system representations \cite{Simonsson2024}. For single-producer networks, researchers have successfully applied convexification of Kirchhoff loop constraints to eliminate nonconvex momentum equations \cite{Agner2022}. 

Regarding MPC, the nonlinear nature of DHNs often complicates the design of a real-time implementable controller. Hence, several linearized formulations have been proposed, see, e.g., \cite{Sandou2005, Verrilli2016, Farahani2016, Quaggiotto2021}. Recently, a nonlinear MPC and a mixed-integer nonlinear MPC were introduced for a \text{small-scale} network~\cite{Jansen2023, Jansen2024}, allowing the neglect of thermal transients to maintain a tractable formulation. Similarly, \cite{Frison2024} developed an MPC algorithm for prosumer DHNs with storage, leaving out thermal transients and heat losses. While computationally attractive, disregarding thermal transients is not suitable for large-scale DHNs due to the significant time delays~\cite{Ben1995}.

A nonlinear MPC scheme that considers thermal transients, multiple producers, and storage was considered in~\cite{Rose2024}. However, this work employed a stabilizing scheme instead of an economic one, did not rigorously address pressure and Kirchhoff loop constraints, and prosumers were not included. Furthermore, \cite{Labella2023} developed a nonlinear MPC for the AROMA network with fixed-volume layered storage, but did not account for multiple producers or prosumers. Alternative MPC-based approaches have emerged through decomposition-based methods, including distributed MPC to separate computation among agents~\cite{Lefebure2022} and temporal decomposition for MPC solved through Benders decomposition~\cite{Sibeijn2024a}. Learning approaches have also emerged, utilizing physics-informed neural networks for DHN modeling~\cite{deGiuli2024} and incorporating support vector machines to handle comfort constraints in MPC formulations~\cite{Khosravi2024}.

While significant advancements have been made in the operational optimization and control of DHNs, current methods often lack consideration of essential features of 4th generation DHNs, such as multiple producers, prosumers, bidirectional flows, and thermal transients. Additionally, many approaches do not consider an economic objective or do not consider the complexity and scalability necessary to manage large-scale networks in real time. These gaps highlight the need for further research to develop comprehensive solutions that integrate these elements and address the computational challenges inherent in MPC for DHNs.

\subsection{Contributions}
In this work, we focus on economic model predictive control of district heating networks. Hence, we consider the problem of scheduling and management for economic operation. We state our contributions as follows:
\begin{enumerate}
    \item We provide a control- and optimization-oriented, graph-theoretic model for DHNs that includes the following features: \emph{multiple producers, prosumers, storage, bidirectional pipe flow, thermal transients, adaptable model resolution}, and a \emph{generalized Kirchhoff loop convexification} approach to accommodate DHNs with all previously mentioned features.
    We note that, to the best of the authors' knowledge, this is the first work to integrate all these features. In particular, we highlight the inclusion of prosumers, formerly only considered without thermal transients in~\cite{Frison2024}, and the generalization of the Kirchhoff loop convexification approach, previously limited to single producer DHNs~\cite{Agner2022}.
    \item We introduce a novel economic nonlinear MPC algorithm for DHNs and we provide numerical analysis on the convergence properties of the proposed controller.
    \item We conduct a comprehensive study into the numerical performance of the proposed control methods in a closed-loop setting, specifically examining the added value of incorporating storage and multiple producers. 
    Additionally, we provide an in-depth analysis of the computational efficiency of our algorithms.
\end{enumerate}

The paper is structured as follows. In Section~\ref{sec:MAth}, we introduce our general approach used to model DHNs. The following two sections are dedicated to specific modeling techniques for the hydraulic system (Section~\ref{sec:hydraulics}) and the thermal system (Section~\ref{sec:thermaldyn}). In Section~\ref{sec:EMPC}, we formulate the economic MPC problem and discuss practical and theoretical convergence properties of the closed-loop system. Lastly, Section~\ref{sec:results} contains the simulation results.

\section{District Heating Network Model} \label{sec:MAth}

In this section, we detail the fundamental modeling choices and conventions, certain parts of which have been introduced in earlier work by the same authors~\cite{Sibeijn2024}. Firstly, we define the graph-theoretic notions used to model the DHN in Section~\ref{ssec:graphmodel}. Subsequently, we discuss the transient thermal dynamics of pipe flow in Section~\ref{ssec:edge_dynamics}, and conservation constraints that are enforced through nodes in Section~\ref{ssec:node_constraints}. Finally, we introduce an example DHN called the AROMA network in Section~\ref{ssec:AROMA}.

\cm{The modeling framework builds upon a fundamental characteristic of DHNs: their symmetric structure consisting of parallel supply and return pipelines. The supply network delivers hot water from heat sources to consumers, while the mirrored return network carries the cooled water back for reheating. This configuration enables cost-effective installation through shared infrastructure while supporting various pipe arrangements for efficient heat distribution \cite{Neale1987,ElMrabet2024}.
}

\subsection{Graph Model} \label{ssec:graphmodel}

The DHN consists of hydraulic and thermal components such as pipes, junctions, pumps, valves, heat exchangers, and buffers. We model this network as a strongly connected directed graph $\mathcal{G} = (\mathcal{N},\mathcal{E})$ with a set of nodes $\mathcal{N}$ that represent junctions in the DHN, which are connected by edges $\mathcal{E}\subseteq \mathcal{N}\times\mathcal{N}$, representing pipelines that may be fitted with pumps, valves, or heat exchangers. Note that strong connectivity of $\mathcal{G}$ ensures no mass exits the system. 

For the ease of discussion, let $\mathcal{N}$ be characterized as $\mathcal{N}:=\{1,2,\ldots,\mj\}$. The \textit{adjacency matrix}
\cm{$D \in \{0,1\}^{\mj \times\mj}$}
of $\mathcal{G}$ describes node-to-node connections, i.e., for any $(i,j)\in\mathcal{N}\times\mathcal{N}$, we have  
\begin{equation}
    D_{ij} = \begin{cases}
        1, \quad  \text{if } (i,j) \in \mathcal{E},\\
        0, \quad \text{otherwise}.
    \end{cases}
\end{equation}
Given an enumeration for the edges of the graph modeling our DHN,
the \textit{incidence matrix}
\cm{$E \in \{-1, 0, 1\}^{\mj \times \me}$} of $\mathcal{G}$ describes the edge-to-node relationship, i.e., for any  $(i,k)\in\mathcal{N}\times\{1,2,\ldots,\me\}$, we have 
\begin{equation}
    E_{ik} = \begin{cases}
        -1, &\quad \text{if the } k^{\text{\tiny{th}}} \text{ edge exits node } i,\\
        1,  &\quad  \text{if the } k^{\text{\tiny{th}}} \text{ edge enters node } i,\\
        0,  &\quad \text{otherwise}.
    \end{cases}
\end{equation}

In our model of the DHN, nodes correspond to volume-less junctions and edges correspond to pipes that may be equipped with heat exchangers, valves, or pumps \cite{Machado2022}. On the other hand, for the sake of simplicity, we consider storage to consist of a fixed volume set of pipe segments, as used in \cite{Labella2023}, rather than the varying volume approach adopted by \cite{Machado2022}. 
Consequently, the dynamic evolution of pressure and temperature within junctions reduces to algebraic equations, with pressure and temperature behavior fully determined by the edge dynamics.

\subsection{Edge Dynamics} \label{ssec:edge_dynamics}
We model the dynamics of an edge representing a pipe, potentially equipped with a heat exchanger, pump, or a valve, through an approximation of the \mbox{1-dimensional} compressible Euler equations \cm{and the thermal energy equation} for cylindrical pipes \cite{Krug2020, Danschel2023}. Hence, the dynamics of edge $e\in\mathcal{E}$ are described through the following PDEs:
\begin{equationarray}{l}
        \partial_t \rho_e + \partial_x (\rho_e v_e) = 0, 
        \label{eq:inc_fluid}\\
        \notag \vspace{-2.5mm} \\
        \partial_t (\rho_e v_e) + \partial_x (\rho_e v_e^2) + \partial_x p_e
        + \rho_e g  \hat{z}_e 
        \notag\\
        \notag\vspace{-4mm}\\
        \qquad\qquad\qquad\qquad\quad\quad
        + K_e \frac{ \rho_e}{2d_e}\,|v_e|v_e = 0, 
        \label{eq:cons_mom}\\
        \notag \vspace{-2.5mm} \\
        \partial_t T_e + v_e\, \partial_x T_e + \frac{p_e}{\rho_e c_p}\partial_x v_e
         - \frac{ K_e}{2c_p d_e}\,|v_e|v_e^2
         \notag\\
        \notag \vspace{-4mm} \\
        \qquad\qquad\qquad\qquad\quad\quad
         +\frac{4U_e}{\rho_e c_p d_e}\left( T_e - T_{a}\right)  = 0 \label{eq:advec_ener},
\end{equationarray}
which are essentially characterizing
the temporal and spatial evolution of three central variables; \textit{temperature} ${T_e}(t,x)$~$[$K$]$,  \textit{flow velocity} ${v_e}(t,x)$~$[\text{m\,s}^{-1}]$, and \textit{pressure} ${p_e}(t,x)$~$[$Pa$]$ of water. The other parameters are the density of water $\rho_e\; [\text{kg\,m}^{-3}]$, gravity $g\; [\text{m\,s}^{-2}]$, slope of pipe $\hat{z}_e \;[\text{-}]$, friction coefficient $K_e\;[\text{-}]$, diameter of pipe $d_e \;[\text{m}]$, heat transfer coefficient $U_e \; [\text{J\,m}^{-2}\text{K}^{-1}]$, specific heat capacity of water $c_p\;[\text{J\,kg}^{-1} \text{K}^{-1}]$, and ambient temperature $T_a\; [\text{K}]$.

Concerning the dynamics of DHNs, it is commonly assumed~\cite{Krug2020, Danschel2023} that the water inside the system is incompressible and has constant density, i.e., we have $\partial_x v_e = 0$, $\partial_t \rho_e = 0$, and $\partial_x \rho_e = 0$, for each $e\in\mathcal{E}$. 
Moreover, heat \cm{generated through} friction is negligible in practice~\cite{Danschel2023}, particularly compared to other terms in \eqref{eq:advec_ener}. 
Therefore, we \cm{omit} the term
$
    \frac{K_e}{2c_p d_e}|v_e|v_e^2
$ from \eqref{eq:advec_ener}.
\cm{Moreover, since pipelines are typically laid underground at constant depth \cite{ElMrabet2024}, we assume no elevation differences throughout the network, i.e., $\hat{z}_e = 0, \,\, \forall \, e \in \mathcal{E}$.}
Additionally, considering the significant separation in time scales between thermal and hydraulic dynamics, and since the frictional term in \eqref{eq:cons_mom} dominates the inertial term \cite{c0}, one can neglect dynamics on the flow rate, i.e., ${\partial_t v_e} = 0$, for any $e\in\mathcal{E}$. 
Thus, for the dynamics of edge $e\in\mathcal{E}$, we have the following 
equations
\begin{equationarray}{l}
        \partial_x p_e+ K_e \frac{ \rho}{2d_e}|v_e|v_e = 0, \label{eq:cons_mom2}\\
        \notag \vspace{-2.5mm} \\
        \partial_t T_e + v_e \partial_x T_e  +\frac{4U_e}{\rho c_p d_e}\left( T_e - T_{a}\right)  = 0 \label{eq:advec_ener2}.
\end{equationarray}

\subsection{Nodal Constraints} \label{ssec:node_constraints}

The nodal constraints follow from conservation laws, i.e., the incompressibility condition implies that mass must be conserved over any node. 
More precisely, for any node $n \in \mathcal{N}$, we define the edge sets 
$\mathcal{E}_{\rightarrow n} = \{e \in \mathcal{E} : e \text{ enters } n\}$ and $\mathcal{E}_{n \rightarrow} = \{e \in \mathcal{E} : e \text{ exits } n\}$ as the set of entering and exiting edges, respectively, and subsequently, mass conservation is characterized through the following flow balance equation 
\begin{equation} \label{eq:mass_cons}
    \sum_{e\in \mathcal{E}_{\rightarrow n}}q_e(t)  = \sum_{e\in \mathcal{E}_{n \rightarrow}} q_e(t), 
\end{equation}
where $q_e(t) = \Phi_e v_e(t)$ is the flow rate in pipe $e$ with $\Phi_e$ the cross-section of pipe. Additionally, energy balance should also be considered for any node, which can be described as a mixing rule determining the relationship between the exit temperature of a node as a function of temperatures of entering flows. 
Accordingly, one can employ a mixing rule that takes a flow-weighted average of the incoming temperatures, as in \cite{Machado2022} and \cite{Krug2020}, to obtain the temperature of a node.
More precisely, for the temperature of any node $n\in\mathcal{N}$, we have
\begin{equation} \label{eq:mixing}
    T_n(t) =  \frac{\sum_{e\in \mathcal{E}_{\rightarrow n}}q_e(t)T_e(t)}{\sum_{e\in \mathcal{E}_{\rightarrow n}}q_e(t)}.
\end{equation}
Also, the temperature of any edge exiting a node is expected to be the same as the node temperature, i.e., we have
\begin{equation} \label{eq:exit_temp}
    T_e(t) = T_n(t), \qquad \forall\, n\in\mathcal{N}, \forall\, e \in \mathcal{E}_{n \rightarrow}.
\end{equation}
The equations presented so far give rise to a graph theoretical model for DHNs. Thus, the system is described through a set of partial differential algebraic equations (PDAEs),
namely \eqref{eq:cons_mom2}~-~\eqref{eq:exit_temp},
which is used as the base formulation for our future derivation in this paper.
It is worth noting the DHN model in the current form is not applicable 
for optimization purposes mainly due to the complex infinite-dimensional nature of the mentioned PDAEs. 
Accordingly, in the following sections, we will develop more tractable and yet realistic models, suitable for optimization-based control strategies.

\subsection{AROMA: A Benchmark for District Heating Networks} \label{ssec:AROMA}

Before proceeding, we introduce the AROMA district heating network, which is widely used in DHN literature as a benchmark example to evaluate the performance of numerical algorithms, see, e.g., \cite{Krug2020, Labella2023}. Therefore, we will employ this network to demonstrate the methodology developed in this paper. The network, depicted in Figure~\ref{fig:AROMA_DIAGRAM}, features multiple consumers, producers and a storage unit. Originally, the AROMA network consisted of only a single producer without any storage included~\cite{Krug2020}. Recently~\cite{Labella2023}, storage was featured in the network by augmenting the AROMA network model. Considering that we are focused on the general case of DHN, namely with multiple producers/prosumers,  for the demonstration of the proposed modeling scheme and the numerical validation of the developed methodology in this paper, we additionally augment the AROMA DHN to accommodate the mentioned extensions and features. Nevertheless, one should note that our methodologies are applicable beyond the AROMA network. 

\begin{figure}[thbp]
    \centering
    \includegraphics[width = \linewidth,trim={7.2cm 8cm 7.2cm 8.2cm}, clip]{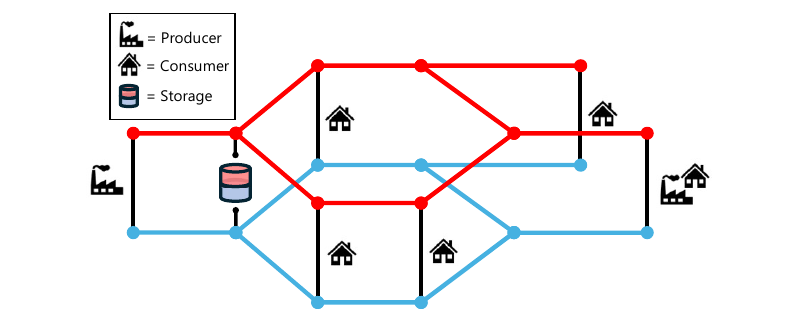}
    \caption{The AROMA district heating network with multiple consumers, a prosumer, a storage buffer, and a loop.}
    \label{fig:AROMA_DIAGRAM}
\end{figure}

\section{DHN Model: A Tractable Reformulation for the Hydraulics} \label{sec:hydraulics}
In this section, we aim to improve the tractability of the DHN model introduced in Section~\ref{sec:MAth}, with a particular focus on the hydraulic dynamics described by the equations \eqref{eq:cons_mom2} and \eqref{eq:mass_cons}. To this end, we first introduce a new graph to accommodate bidirectional flows. Secondly, we describe the method used to compute a set of independent cycles that fully describe the mass flowing through the network. Thirdly, we address the treatment of the momentum equation and introduce a convex reformulation for these constraints.

\subsection{Fixing the Flow Direction}
Considering the graph abstraction of a district heating network as in Figure~\ref{fig:AROMA_DIAGRAM}, each pipeline is represented by a string of edges connected in series. Thus, in the coarsest representation of our system, we have exactly one edge for each pipeline until it reaches a junction. 
Hence, we have $\me$ pipelines, where with respect to each $e\in\mathcal{E}$, we define flow variables $q_e$ that belong to the interval $\mathcal{Q}_e \coloneqq \{q_e : \underline{q}_{e} \leq q_e \leq \overline{q}_e\}$, \cm{where $\overline{q}_e > 0$}.

While seemingly a relatively minor detail, allowing  the direction of the flow to switch during the operational phase has significant implications for the modeling procedure and complicates the design of a tractable controller. Therefore, for the edges where directional switching is allowed, i.e., for $e\in\mathcal{E}$ with $\underline{q}_e < 0$ and $\overline{q}_e > 0$, we decompose the corresponding flow variable as 
\begin{align} \label{eq:complementary_con}
    q_e = q_e^+ - q_e^- \; \; \text{with} \; \;q_e^+,q_e^- \geq 0 \; \;\text{and}\; \; q_e^+  q_e^- = 0.
\end{align}
In terms of the graph, we obtain a new graph $\mathcal{G}^+~=~(\mathcal{N}, \mathcal{E}^+)$ with additional set of edges $\mathcal{E}^+ = \mathcal{E} \cup \Delta_{\mathcal{E}}$, where \cm{$\Delta_{\mathcal{E}} \coloneqq \{ (v, u) \notin \mathcal{E} \mid e = (u, v) \in \mathcal{E} \text{ and } \underline{q}_{e} < 0 \}$}. To illustrate, we depict $\mathcal{G}^+$ for the AROMA network in Figure~\ref{fig:graph_aroma}. Here, $\mathcal{G}^+$ is directed, and double-sided arrows indicate pipelines that allow bidirectional flow. The direction matching the original orientation in $\mathcal{G}$ will have the flow $q_e^+$ associated with it, while the newly added reverse edge has $q_e^-$ associated with it.

\begin{figure}[thbp]
    \centering
    \includegraphics[width = .72\linewidth,trim={7cm 7.65cm 7cm 7.65cm}, clip]{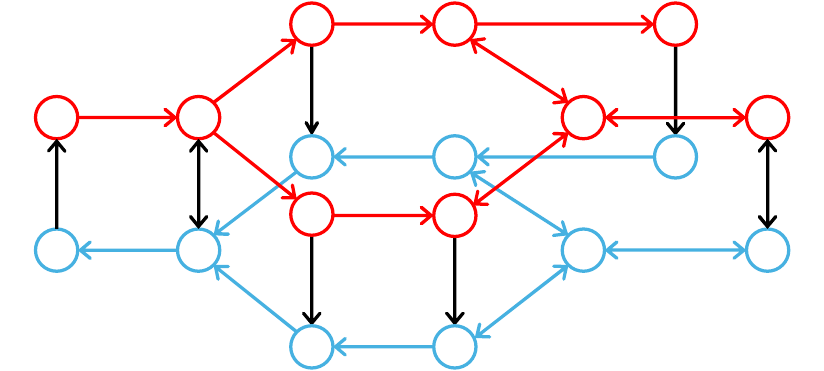}
    \caption{A directed graph abstraction of the AROMA district heating network with fixed flow directions. Double-sided arrows indicate two opposite facing edges.}
    \label{fig:graph_aroma}
\end{figure}

\begin{remark} \normalfont
    \cm{In typical DHNs, bidirectional flow is neither necessary nor practical for all pipeline segments. \cm{However, bidirectional flow capability is essential for enabling prosumer operation in the network.} The determination of flow directionality is made during the design phase through systematic analysis of network architecture.
    Therefore, not all edges in Figure \ref{fig:graph_aroma} are bidirectional. 
    \hfill$\diamond$}
\end{remark}
\subsection{Independent Flows}

We denote the vector of edge flow rates as $q = [q_e]_{e \in \mathcal{E}^+}$. The flow on an edge depends linearly on the flow on all other edges in the network due to conservation of mass. As a result, we can reduce the number of free flow variables that we need to optimize for. To this end, we introduce the \textit{reduced loop matrix} $F_r$ which maps the reduced flow vector $q_r \in \R_{+}^{\mr}$ to $q$ through
\begin{equation} \label{eq:id_flows}
    q = F_r^\top q_r,
\end{equation}
with $\mr < \mep$.
\cm{Intuitively, the vector $q_r$ represents the flow which circulates the network, meaning that it passes through a supply section and its mirrored return section reaching back to its starting point. In the process, these flows pass through producers, consumers, or storages.}
Additionally, we introduce the \textit{fundamental loop matrix} $F$ which maps all fundamental flows $q_f \in \R^{\mf}$ to $q $, i.e.,
\begin{equation} \label{eq:id_flows2}
    q = F^\top q_f,
\end{equation}
with $\mf \leq \mr$. \cm{The vector $q_f$ depends linearly on the elements of $q_r$}.
The matrices $F_r$ and $F$ consist of columns that describe directed cycles within $\mathcal{G}^+$. Nonetheless, $F$ is full column rank while $F_r$ may not be. Similarly, $q_r$ is a vector in the positive orthant, whereas the entries of $q_f$ are not required to be nonnegative. 
The reason for introducing both $F_r$ and $F$ is that the nonnegativity of $q_r$ in \eqref{eq:id_flows} significantly improves the numerical results compared to using \eqref{eq:id_flows2},  while $F$ is primarily employed for theoretical purposes, as further detailed in Section~\ref{ssec:khoff_loops}.

In \cite{Persis2011}, a method is presented to compute the fundamental loop matrix $F$  by setting the free flow variables as the flows through the chords of the spanning tree of $\mathcal{G}$. Consequently, a fundamental loop is defined as the loop that is formed whenever a chord is re-connected to the spanning tree. Then, the fundamental loop matrix $F$ has elements \cm{$F_{ij} \in \{-1,0,1\}$, for all $i$ and $j$,} depending on the orientation of the chord and whether an edge is part of a fundamental loop.

To preserve nonnegativity of all elements in $F$, we develop a slightly different approach to compute $F$.
To this end, we need the notion of a \emph{directed cycle} \cite{Douglas2001}.
\begin{definition}[Directed cycle] \normalfont
    A directed cycle $\mathcal{C}$ in $\mathcal{G}$ is a sequence of nodes and edges as $\left(n_0, e_1, n_1, e_2, \dots, e_k, n_k \right)$ such that 
    \begin{enumerate}
        \item $n_0,n_1,\dots,n_k$ are different nodes,
        \item $e_i = (n_{i-1}, n_i)$ for $i = 1,\dots, k$,
        \item $n_0 = n_k$,
    \end{enumerate}
    \cm{where two sequences are an equivalent cycle if one can be obtained from the other by a cyclic permutation.}
\end{definition}
Let $\mathcal{S}(\mathcal{G}) \coloneqq \{\mathcal{C}_i(\mathcal{G})\}$ be the \cm{set} of all directed cycles within $\mathcal{G}$. The matrix $F_r$ and $F$ can be computed through the following procedure:
\begin{enumerate}
    \item Compute all directed cycles $\mathcal{S}(\mathcal{G}^+)$ in $\mathcal{G}^+$.
    \item Remove any cycles that consist of at most two nodes.
    \item Remove any asymmetric cycles, i.e., cycles for which their path is not mirrored between the supply and return network. 
   We call this reduced set $\mathcal{S}_r (\mathcal{G}^+) $.
    \item Define reduced loop matrix $F_r$ with elements
    \begin{equation*}
        F_{r,ij} = \begin{cases}
            1, \quad \text{if edge } j \in \mathcal{C}_i  \text{ and cycle } \mathcal{C}_i \in \mathcal{S}_r (\mathcal{G}^+),\\
            0, \quad \text{otherwise},
        \end{cases}
    \end{equation*}
    where $F_r$ has $\mr = |\mathcal{S}_r (\mathcal{G}^+)|$ rows and $\mep$ columns.
    \item Compute $F$ by performing a pivoted QR factorization on $F_r^\top$ as  
    \begin{equation}
        F_r^\top P = QR,
    \end{equation}
    where $R \in \R^{\mep \times\, \mr}$ is an upper triangular matrix, $Q\in\R^{\mep \times \mep}$ is an orthonormal matrix, $P\in\R^{\mr \times \mr}$ is a permutation matrix, and the columns of $F_r^\top P$ are projected onto an orthonormal basis spanned by the columns of $Q$. If rank($F_r$) = $\mf < \mr$, all elements $R_{ii}$ with $i > \mf$ are zero, indicating that corresponding columns of $F_r^\top P$ lie in the subspace of all prior columns. Hence, we obtain the fundamental loop matrix by taking only the first $\mf$ columns of $F_r^\top P$ as follows:
    \begin{equation}
        F = \left( F_r^\top P \begin{bmatrix}
            I_{\mf} \\ \mathbf{0}
        \end{bmatrix} \right)^\top.
    \end{equation}
   The set of remaining cycles is denoted as $\mathcal{S}_f(\mathcal{G}^+)$.
\end{enumerate}
\cm{Matrix $F$, constructed via QR factorization with column pivoting on $F_r^\top$, inherits nonnegativity from $F_r$ and possesses full column rank, ensuring mass conservation within the network.}

\begin{remark}\normalfont
    We note that, in principle, the rows of the computed fundamental loop matrix span the same basis as the row space of the matrices introduced in \cite{Persis2011} and \cite{Machado2022}. The only difference is that, in this context, we exclude asymmetric cycles without loss of generality.
    \hfill$\diamond$
\end{remark}

\subsection{Treatment of the Momentum Equation}
In practice, due to the significant time scale separation between the hydraulic and thermal dynamics, the control of corresponding components within the DHN is executed over different time intervals.
Thus, the dynamic behaviour of flows and pressure throughout the network is less of concern here.
Nonetheless, one needs to certify feasibility for the hydraulic operation. More precisely, we want to identify a set $\mathcal{Q}$, such that for all $q \in \mathcal{Q}$, pressure remains within the limits at all nodes of the network, and Kirchhoff's second law, stating that the sum of pressure differences along each loop in the network equates to zero, is satisfied. 

One can approximate the conservation of momentum equation \eqref{eq:cons_mom2} by substituting $v_e$ with $v_e = \frac{4q_e}{\pi d_e^2}$ and discretizing $\partial_x p = \frac{\Delta p_e}{L_e}$, \cm{where $L_e$ denotes the pipe length}, for each $e\in\mathcal{E}^+$. Accordingly, one obtains the equation describing pressure drop over pipe segment $e$ caused by friction as
\begin{equation} \label{eq:pipe_pdrop}
    \Delta p_e = 8 \rho L_e \frac{K_e}{\pi^2 d_e^5} |q_e|q_e = R_{\mu,e} q_e^2, \qquad\forall\, e \in \mathcal{E}^+,
\end{equation}
where $R_{\mu,e}$ is a combined constant term representing the frictional resistance in the pipe. Additionally, the sign dependency for $q_e$ is removed according to \eqref{eq:complementary_con}, i.e.,  the flow direction on each edge is fixed.

We define the sets $\mathcal{P} \subset \mathcal{E}^+$ and $\mathcal{V} \subset \mathcal{E}^+$ that contain the edges with a pump and a valve, respectively. Then, for each edge $e \in \mathcal{P} $, the pressure change over $e$ is described by
\begin{equation} \label{eq:pump_pdrop}
\begin{aligned}
    \Delta p_e &= R_{\mu,e}q_e^2 - h_e(r_e),
\end{aligned}
\end{equation}
with $h_e(r_e) = c_e r_e$ being the \cm{pressure difference} induced by the pump, $c_e$ being the maximum pumping power capacity, and  $r_e \in [0, 1]$ being the \cm{normalized pump speed}. Additionally, for each edge $e \in \mathcal{V}$,
the pressure drop over $e$ is described by
\begin{equation} \label{eq:valve_pdrop}
    \Delta p_e = R_{\mu,e}q_e^2 + R_{\nu,e}(\nu_e) q_e^2,
\end{equation}
where $0 \leq R_{\nu,e}(\nu_e) < \infty$ is a time-varying control variable and $\nu_e \in [0,\infty)$ is the position of the valve such that $R_{\nu,e}(0) = 0$ corresponds to the full opening of the valve and $\lim_{\nu_e \rightarrow \infty}R_{\nu,e}(\nu_e) \rightarrow \infty$ corresponds to the full closing of the valve. Here, we are not considering specific valve types and characteristics. It is sufficient to only assume that any increase in $\nu_e$ will increase the resistance over the valve \cm{such that $R_{\nu,e}$ is strictly monotone.} Nonetheless, for the ease of discussion, one can consider $R_{\nu,e}$ to be linear in $\nu_e \in [0, \infty)$, i.e., we have $R_{\nu,e}(\nu_e) = R_{\nu,e}\nu_e$, where $R_{\nu,e}$ is a positive scalar indicating the resistance coefficient of the valve.

\subsection{Convex Reformulation of the Kirchhoff Loop Constraints} \label{ssec:khoff_loops}
The constraints imposed by the equations \eqref{eq:pipe_pdrop}-\eqref{eq:valve_pdrop} are nonconvex with respect to the flow rate due to the quadratic friction terms, and therefore, they need to be reformulated. 
More precisely, defining $\mathcal{S}^\dagger(\mathcal{G}^+)$ as the set of \textit{all} cycles in $\mathcal{G}^+$ 
independent of the direction of the edges, 
Kirchhoff's second law states that 
\begin{equation} \label{eq:Khoff}
    \sum_{e \in \mathcal{C}_{i}} \Delta p_j = 0, \quad \quad \forall \mathcal{C}_i \in \mathcal{S}^\dagger(\mathcal{G}^+),
\end{equation}
i.e., for each cycle, we need to satisfy a quadratic \cm{equality} constraint, which is essentially nonconvex. 

The pressure is defined pointwise in space, and therefore, we consider a pressure variable $p_n$, for each $n \in \mathcal{N}$. Let $\mathbf{p}_\mathrm{n} \in \mathbb{R}^{\mj}$ be the vector of nodal pressures and $\Delta \mathbf{p}_\mathrm{e} \in \mathbb{R}^{\mep}$ be the vector of pressure differences over all edges. Combining \eqref{eq:pipe_pdrop}-\eqref{eq:valve_pdrop}, we have
\begin{equation} \label{eq:delta_p}
    -E^\top \mathbf{p}_\mathrm{n} = \Delta \mathbf{p}_\mathrm{e} = {R}(\nu) \left(q \odot q \right) - {H}(r),
\end{equation}
where $\odot$ denotes the element-wise product of two vectors,
$R(\nu) = R_\mu + R_\nu(\nu)$ is a diagonal matrix that has on its diagonal the sum of resistances due to friction and valve effects on each edge,
and $H(r)$ is the vector of induced pump \cm{pressure difference} for each edge. 
As in \cite{Machado2022}, we multiply \eqref{eq:delta_p} from the left by $F$ to obtain the sum of pressure difference over all fundamental cycles as
\begin{equation} \label{eq:basis_cycle_pressure}
    -F E^\top \mathbf{p}_\mathrm{n} = \sum_{e \in \mathcal{C}_{i}} \Delta p_e = 0, \quad \quad \forall \mathcal{C}_i \in \mathcal{S}_f(\mathcal{G}^+).
\end{equation}
Subsequently, combining \eqref{eq:delta_p} and \eqref{eq:basis_cycle_pressure} leads to
\begin{equation}  \label{eq:equality}
    FR(\nu)(q\odot q) = FH(r),
\end{equation}
implying that the sum of all pressure drops due to valve effects and friction in any directed cycle of the graph should be equal to the sum of all induced pump \cm{pressure differences} within that same cycle. Similarly, we introduce the inequality
\begin{equation} \label{eq:inequality}
    F_r{R}_\mu(q\odot q) \leq F_rH(\mathbf{1}), 
\end{equation}
where $\mathbf{1}$ denotes the vector of all ones, i.e., $\mathbf{1} = [1,\,\dots,\,1]^\top$. Equation \eqref{eq:inequality} represents an upper bound on the induced \cm{pressure difference} along each loop and corresponds to the scenario where all valves are maximally open and all pumps are operated at maximum capacity, i.e.,  $\nu_e = 0, \, \forall \, e\in \mathcal{V}$, and $r_e = 1,\, \forall \, e\in \mathcal{P}$.

Before proceeding to the main result of this section, we need to introduce the following assumption.

\begin{assumption} \normalfont\label{ass:unique_valves}
   There are 
   $|\mathcal{V}| \geq \mf$ edges with a valve placed throughout the network, and the position of these edges is such that rank($F\Pi) = \mf$, where 
   $\Pi \in \{0, 1\}^{\mep \times |\mathcal{V}|}$ is a selection matrix with entries
   \begin{align*}
       \Pi_{ek} = \begin{cases}
           1, \quad \text{if valve $k$ is on edge $e$},\\
           0, \quad \text{otherwise},
       \end{cases}
   \end{align*}
   for any $e\in\{1,2,\ldots,\mep\}$ and any $k\in\{1,2,\ldots,|\mathcal{V}| \}$.  
   Since $F\Pi$ is full row rank, it can be mapped back to $F$ through a linear transformation, 
   i.e., we have 
   \begin{equation}
       F = F\Pi \Psi,
   \end{equation}
   where $\Psi = (F\Pi)^\dagger F$ and $(\cdot)^\dagger$ denotes the Moore-Penrose pseudo-inverse.
   Let $m_\Theta$ denote the dimension of the null space of matrix $F\Pi$.   
   We assume that there exist matrices $Z_1, Z_2 \in \R^{m_\Theta \times \mep}$
   such that
   \begin{equation} \label{eqn:ass1}
       (\Psi + \Theta_{F\Pi}Z_1)H(\mathbf{1}) \ge (\Psi + \Theta_{F\Pi}Z_2)R_\mu(q \odot q)
   \end{equation}
   where $\Theta_{F\Pi} \coloneqq I - (F \Pi)^\dagger F\Pi$ denotes the kernel of $F\Pi$. 
\end{assumption}

Before proceeding further, we need to highlight several remarks about the introduced assumption.

\begin{remark} \normalfont \label{rem:assumption1}
Multiplying \eqref{eqn:ass1} from the left by $F_r\Pi$ yields \eqref{eq:inequality}. In the majority of cases, the validity of \eqref{eq:inequality} implies that \eqref{eqn:ass1} is also satisfied. Nonetheless, \eqref{eqn:ass1} can be explicitly guaranteed by imposing it as a constraint within the MPC formulation introduced in Section~\ref{sec:EMPC}. Note that, to ensure convexity, it may be required to determine $Z_2$ a priori such that 
$\Psi + \Theta_{F\Pi} Z_2$ is a matrix with nonnegative entries.
\hfill$\diamond$
\end{remark}

\begin{remark} \normalfont \label{rem:valve_placement}
One can easily ensure the existence of matrix $Z_2$ such that the right-hand side of \eqref{eqn:ass1} is element-wise nonnegative. 
To this end, we need to consider only the supply section of the DHN. Then, the corresponding graph consists of source nodes coming from producers, sink nodes reaching consumers, and intermediate nodes that represent junctions. Additionally, we consider networks where the degree of each intermediate node is at most three, meaning that each junction is either a \textit{splitting node} or a \textit{merging node}. 
Valves are positioned in reverse cascading order from consumers to intermediate nodes. At splitting nodes, valves are required on outgoing edges, except when an edge connects directly to another splitting node. For merging nodes, valve placement is only necessary on edges originating from producer nodes; otherwise, the analysis proceeds to the subsequent node.
This process is repeated until all edges are traversed. We illustrate this procedure in Figure~\ref{fig:remark_fig}. The rationale for this valve checking strategy is that it ensures the flow on any edge lacking a valve becomes a nonnegative linear combination of the flows on edges equipped with valves. Formally, this implies the existence of a matrix $\Psi + \Theta_{F\Pi} Z_2$ with exclusively nonnegative entries. 
\hfill$\diamond$\\
\cm{\textbf{Remark to practitioners.} While the valve placement strategy may appear extensive, it reflects the fundamental requirement that a high degree of freedom in flow control requires a sufficiently high number of strategically placed pumps or valves. 
In practical situations, where valve placement is restricted by physical or economic constraints, analyzing \eqref{eqn:ass1} helps determine where fewer valves can still preserve convexity.
In cases where achieving convexity is not feasible, it serves as a practical tool for practitioners to uncover inherent system limitations. These insights can support the design phase by identifying the most impactful valve placements or guide the development of operational constraints that satisfy \eqref{eqn:ass1} within practical limitations.}
\hfill$\diamond$
\end{remark}



\begin{figure}[thbp]
\centering
\includegraphics[width=\linewidth, trim={6.5cm 7.4cm 6.5cm 7.4cm}, clip]{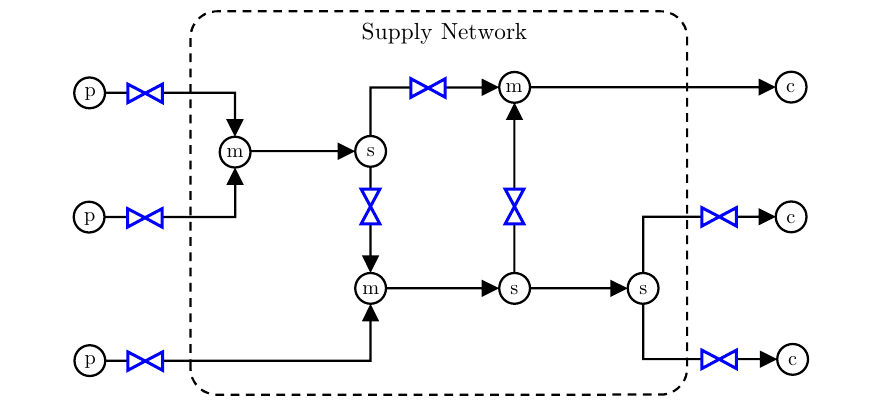}
\caption{The valve checking procedure, described in Remark~\ref{rem:valve_placement}, illustrated by showing the flow from producers (p) to consumers (c) going through the supply network consisting of merging nodes (m) and splitting nodes (s). }
\label{fig:remark_fig}
    
\end{figure}

\begin{remark} \normalfont
\cm{The valve placement approach described in Remark \ref{rem:valve_placement} extends naturally to networks with bidirectional edges, where nodes connect to edges that can experience flow in both directions are treated as both merging and splitting nodes to account for all operational modes. Nevertheless, for bidirectional edges, a single physical valve suffices since only one flow direction is active at any given time.} 
Furthermore, one should note that the introduced assumption on the total number of required valves aligns with the literature on DHNs~\cite{Persis2011, Machado2022}, where each chord of the graph's spanning tree is equipped with a valve. 
\hfill$\diamond$
\end{remark}

The following proposition guarantees that \eqref{eq:inequality} is a sufficient condition for \eqref{eq:equality}.
%
\begin{proposition} \label{prop} \normalfont Let Assumption \ref{ass:unique_valves} hold. 
Then, for any $q$ that satisfies \eqref{eq:inequality}, there exist  ${\nu} \in \R_+^{|\mathcal{V}|}$ and ${r} \in [0, 1]^{|\mathcal{P}|}$ such that \eqref{eq:equality} is satisfied.
\end{proposition}
\vspace{0mm}
\begin{proof}
    %
    %
    Recall that, for any valve openings $\nu$, we have $R(\nu) = R_\mu + R_\nu(\nu)$, where $R(\nu)$ is a diagonal matrix that has on its diagonal the sum of resistances due to friction and valve effects on each edge. Hence, we have
    \begin{equation}
    FR(\nu)(q \odot q) 
    = FR_\mu (q \odot q) + FR_\nu(\nu)(q \odot q).
    \end{equation}
    Furthermore, we know that, for any $\nu$, $q$, and $\Pi$, there exists a positive semi-definite diagonal scaling matrix $D_{R,q}$ and a vector $y = D_{R,q} \nu$, where $y$ represents a scaled substitute of the valve openings, such that we can write $R_\nu(\nu)({q} \odot {q}) = \Pi y$. Therefore,  we have
    \begin{equation}
    FR(\nu)(q \odot q) 
    = FR_\mu (q \odot q) + F\Pi  y.
    \end{equation}
    Accordingly, to show that \eqref{eq:equality} holds for some $y$ and $r$, we need to verify the same argument for the following equation
    \begin{equation}
    FH(r) 
    = FR_\mu (q \odot q) + F\Pi y.
    \end{equation}
    {If we set $r = \mathbf{1}$, then it is enough to show that there exists $y \in [0, \infty)^{|\mathcal{V}|}$ such that 
    \begin{equation} \label{eq:p1_extra}
        F\Pi  y = F\left(H(\mathbf{1}) - R_\mu(q \odot q) \right).
    \end{equation}
    From \eqref{eq:inequality}, we have that the right-hand side of \eqref{eq:p1_extra} is element-wise nonnegative, which is necessary for the existence of nonnegative solutions to $y$. Following Assumption \ref{ass:unique_valves}, 
    we know that \eqref{eq:p1_extra} can be written as
    \begin{equation} \label{eq:linear_system}
        F\Pi (y - \Psi H(\mathbf{1}) + \Psi R_\mu(q\odot q)) = 0.
    \end{equation}
    Hence, the solutions to \eqref{eq:linear_system} lie inside the null space of $F\Pi$. Accordingly, we know that, for each $y$ satisfying \eqref{eq:linear_system}, there exists $z \in \R^{m_\Theta}$ such that
    \begin{equation} \label{eq:p1_extra2}
    \begin{aligned}
        y &= \Psi H(\mathbf{1}) - \Psi R_\mu(q\odot q) + \Theta_{F \Pi} z,\\
        & = (\Psi + \Theta_{F\Pi}Z_1)H(\mathbf{1}) - (\Psi + \Theta_{F\Pi}Z_2)R_\mu(q \odot q),
    \end{aligned}
    \end{equation}
    where 
    $Z_1,Z_2 \in \R^{m_\Theta \times \mep}$ are the matrices introduced in Assumption~\ref{ass:unique_valves}.
    Thus, due to \eqref{eqn:ass1}, we have that $y$ is a vector with nonnegative entries, which implies that $\nu \geq 0$. More precisely, there exist  ${\nu} \in \R_+^{|\mathcal{V}|}$ and ${r} \in [0, 1]^{|\mathcal{P}|}$ such that \eqref{eq:equality} is satisfied. This concludes the proof.
    }
    %
    %
    %
    %
    %
    %
    %
\end{proof}

The previous proposition guarantees, for any $q$ satisfying \eqref{eq:inequality}, the existence of $\nu$ and $r$ satisfying \eqref{eq:equality}, and, therefore, satisfying \eqref{eq:basis_cycle_pressure}. In the following proposition, we show that \eqref{eq:basis_cycle_pressure} implies \eqref{eq:Khoff}, meaning that Kirchhoff's second law is fulfilled.
\begin{proposition} \normalfont \label{prop2}
   Under Assumption~\ref{ass:unique_valves}, we have that \eqref{eq:basis_cycle_pressure} is a necessary and sufficient condition for \eqref{eq:Khoff}, i.e., if there exists a $(q,\nu,r)$ such that \eqref{eq:basis_cycle_pressure}  holds, then \eqref{eq:Khoff} is also satisfied. 
\end{proposition}
\begin{proof}
    Considering $\mathcal{S}_f \subseteq \mathcal{S}^\dagger$, from the definition of \eqref{eq:Khoff} and \eqref{eq:basis_cycle_pressure}, we know that \eqref{eq:Khoff} implies \eqref{eq:basis_cycle_pressure}. 
    Therefore, we only need to prove the sufficiency part of the claim.
    
    Let $F_i$ be the $i^{{\text{th}}}$ row of $F$, i.e., $F_i$ is a vector with elements equal to $F_{ij} = 1$ if  $j \in \mathcal{C}_i$ and $0$ otherwise. All cycles $\mathcal{C}_1, \mathcal{C}_2, \dots, \mathcal{C}_{\mf}$ form a cycle basis for $\mathcal{G}^+$, which means that there are $\mf$ linearly independent basis vectors $F_1, \dots F_{\mf}$ spanning the basis of all other cycles in the graph. As a result, any vector $F^{(\mathcal{C})}$, which corresponds to a cycle $\mathcal{C} \in \mathcal{S}^\dagger \backslash \mathcal{S}_f$, can be constructed by $F_1, \dots F_{\mf}$ as the following integer linear combination
    \begin{equation}
    F^{(\mathcal{C})} = \sum_{i=1}^{\mf} a_i F_i,
    \end{equation}
    where $a_1,\ldots,a_{\mf}\in\mathbb{Z}$.
    Therefore, for the sum of pressure differences over $\mathcal{C}$, we have 
    \begin{equation} 
    \begin{aligned}
    \sum_{e \in \mathcal{C}} \Delta p_e &= F^{(\mathcal{C})} \Delta \mathbf{p}_\mathrm{e}\\
    &= \sum_{i=1}^{\mf} a_i F_i \Delta \mathbf{p}_\mathrm{e}
    = \sum_{i=1}^{\mf} a_i \sum_{e \in \mathcal{C}_{i}} \Delta p_e.
    \end{aligned}
    \end{equation}
    Subsequently, from \eqref{eq:basis_cycle_pressure}, it is implied that
    \begin{equation} 
    \sum_{e \in \mathcal{C}} \Delta p_e = 0,
    \end{equation}
    which concludes the proof.
\end{proof}
In the following proposition, we present the key result regarding the convexity of \eqref{eq:inequality}.
\begin{proposition} \label{prop3}
    The equation \eqref{eq:inequality} is convex with respect to flow vectors $q$ and $q_r$.
\end{proposition}
\begin{proof}
For $i = 1, \dots, \mf$, let matrix $\mathbf{R}_\mu^i$ be defined as 
\begin{equation}
    \mathbf{R}_\mu^i = \mathrm{diag}(F_{r,i})R_\mu,    
\end{equation}
which is a diagonal matrix with nonnegative entries, and thus, positive semidefinite.
According to the definition of $F_r$, $R_\mu$, $h_e$, and equation \eqref{eq:inequality}, we have 
\begin{equation} \label{eq:prop3_eq}
   q^\top \mathbf{R}_\mu^i  \, q  
   =
   \sum_{e \in \mathcal{C}_i} R_{\mu,e}q_e^2
   \leq 
   \sum_{e \in \mathcal{C}_i}  h_e(1) 
   =
   \sum_{e \in \mathcal{C}_i}  c_e, 
\end{equation}
for any $i = 1, \dots, \mr$, which is equivalent to \eqref{eq:inequality} and implies that it is convex with respect to $q$.
For any $i = 1, \dots, \mr$, define matrix $\mathbf{Z}^i$ as
\begin{equation}
    \mathbf{Z}^i = F_r \mathbf{R}_\mu^i F_r^\top. 
\end{equation}
Note that the positive semidefiniteness of $\mathbf{R}_\mu^i$ implies the same property for $\mathbf{Z}^i$, for each $i = 1, \dots, \mr$.
From $q = F_r^\top q_r$, one has 
\begin{equation}
q_r^\top \mathbf{Z}^i  q_r = (F_r^\top q_r) ^\top \mathbf{R}_\mu^i (F_r^\top q_r) = q^\top \mathbf{R}_\mu^i \, q.
\end{equation}
Accordingly, we can write \eqref{eq:prop3_eq}, or equivalently \eqref{eq:inequality}, as
\begin{equation} \label{eq:matKhoff}
    q_r^\top \mathbf{Z}^i  q_r \leq \sum_{e \in \mathcal{C}_i} c_e, \ \quad \forall i = 1, \dots, \mr,
\end{equation}
which implies the convexity of \eqref{eq:inequality} with respect to $q_r$.
This concludes the proof.
\end{proof}

The proposed propositions indicate the presence of a convex reformulation of Kirchhoff's second law through \eqref{eq:inequality}. Proposition~\ref{prop} shows that this reformulation ensures the existence of a \cm{feasible set of valve openings} necessary to achieve the specified flow vector $q$. Moreover, Proposition~\ref{prop2} describes the sufficiency of \eqref{eq:basis_cycle_pressure} for satisfying \eqref{eq:Khoff}. Lastly, we demonstrate the convexity of \eqref{eq:inequality} through a straightforward transformation as detailed in Proposition~\ref{prop3}.

\begin{remark} \normalfont
    \cm{More generally, we expect the results to hold for other convex pressure-flow relationships, e.g., $\Delta p = \mathcal{R}_\mu(q)$, not just the quadratic one typical for turbulent flow in DHNs. When replacing \eqref{eq:pipe_pdrop} with such a function, the optimization problem preserves convexity, since the left-hand side of \eqref{eq:inequality} remains a sum of convex functions. Since Propositions~\ref{prop} and \ref{prop2} are independent of the specific form of $\mathcal{R}_\mu(q)$, the arguments should still apply.
    \hfill$\diamond$}
\end{remark}

\section{DHN Model: A Tractable Reformulation for the Thermal Dynamics} \label{sec:thermaldyn}
To obtain a tractable formulation for the thermal dynamics in the network, a suitable discrete spatial approximation of \eqref{eq:advec_ener2} is required, for which we apply an upwind scheme. Therefore, the dynamics of the $i^{{\text{th}}}$ finite volume cell of water can be described by the following scalar continuous-time ordinary differential equation~(ODE):
\begin{equation}\label{eq:1}
\begin{aligned}
    V_i\Dot{T}_i &= -q_i(T_i - T_{i-1}) - \alpha_i (T_i - T_{a}) + w_i,\\
    y_i &= q_i(T_i - T_a),
\end{aligned}
\end{equation}
where $V_i = \Phi_i \tau_{x_i}$ denotes the cell {volume}, $q_i$ denotes the mass flow of water, $T_i$ denotes the temperature of water in the cell, $T_{i-1}$ is the temperature of the inflow into the cell, $\alpha_i = {4U_i V_i}/{\rho c_p d_i}$ is the heat loss coefficient, $T_a$ is the ambient temperature, $w_i$ is a variable denoting the transfer of heat from or to the environment, and $y_i$ is the output which is proportional to the exergy. Note that, similarly to \cite{Machado2022}, when the cell represents a heat exchanger, $w_i$ indicates the transfer of heat from one side to the other side, and otherwise, the term $w_i$ can be dropped from \eqref{eq:1}. Hence, $w_i$ is a control variable if it corresponds to a controllable producer, and a disturbance if it corresponds to an uncontrollable producer or consumer.

We define the state variable $x_i$ as $x_i = T_i - T_a$, assuming that the ambient temperature is equal and constant for all cells, and consider the corresponding dynamics described as
\begin{equation}\label{eq:dynamics}
\begin{aligned}
    V_i\dot{x}_i &= -(q_i+\alpha_i)x_i + u_{i} +w_i,\\
    y_i &= q_ix_i,
\end{aligned}
\end{equation}
Note that \eqref{eq:1} is equivalent to \eqref{eq:dynamics} when $u_{i} = q_i x_{i-1}$. From here on, we consider any finite volume cell as a node, or more precisely, a thermal node (TN), in the graph. The thermal node has a compartmental structure as shown in Figure~\ref{fig:thermal_node}. 

\begin{remark} \normalfont
    \cm{The energy transfer rate $w_i$ and flow rate $q_i$ in equation \eqref{eq:dynamics} appear decoupled at first glance, however, their relationship becomes evident under steady-state conditions ($V_i = 0$), where the energy balance shows that $w_i$ increases with $q_i$, leading to greater heat transfer at higher flows.
     \hfill$\diamond$
    }
\end{remark}

\subsection{State-space Representation of the DHN Interconnection}
To construct the model of the thermal system, we approximate the spatial evolution of temperature along pipelines, described by \eqref{eq:advec_ener2}, through a finite sequence of partitions. We adopt an approach similar to \cite{Simonsson2024}, where nodes represent volumes of water and edges represent flow rates. 

\begin{figure}[htbp]
    \captionsetup[subfigure]{labelfont={normalfont,small},textfont={normalfont,footnotesize} } 
    \subfloat[\normalfont \cm{The thermal node.}]{
        \raisebox{0.75\height}{\includegraphics[width=0.185\textwidth,trim={12.5cm 10cm 12.5cm 10cm}, clip]{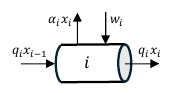}}
        \label{fig:thermal_node}
    }
    \hspace{0.01\textwidth}
    \subfloat[\normalfont \cm{Refinement of the AROMA model.}]{
        \includegraphics[width=0.25\textwidth,trim={11.4cm 9cm 12cm 8.3cm}, clip]{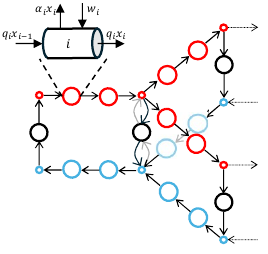}
        \label{fig:oversampling}
    }
    \caption{\cm{Illustration of thermal node model and mesh refinement of the DHN graph. (a) shows the conceptual thermal node model with its dynamics defined by \eqref{eq:dynamics}, while (b) demonstrates the injection of thermal nodes in a section of the AROMA network.}}
    \label{fig:312}
\end{figure}

\subsubsection{Updating the graph}
We introduce intermediate nodes on all edges of the original graph to increase the granularity of the model and improve the approximation of \eqref{eq:advec_ener2}. These intermediate nodes adhere to the compartmental dynamical structure as described in \eqref{eq:dynamics}. Hence, the full model consists of states representing the node-based approximation of pipeline dynamics connected to each other through junction nodes that enforce conservation and mixing constraints introduced in \eqref{eq:mixing} and \eqref{eq:exit_temp}.
More precisely, we consider the original graph $\mathcal{G}^+ = (\mathcal{N}, \mathcal{E}^+)$ with adjacency matrix $D(\mathcal{G}^+) = D^+$. We introduce a vector $l_x \in \mathbb{R}^{\mep}$ that contains the number of additional nodes being inserted on each edge. The new graph is then augmented on the original graph, i.e., we introduce $\mc = \sum_{i=1}^{\mep} l_{x,i}$ new nodes and remove all original edges from $\mathcal{G}^+$.
The new connections are added to the adjacency matrix $\tilde{D}$ with dimension $\mj + \mc$ as follows:
\cmt{\begin{enumerate}
    \item We consider all original nodes to be labeled by their order in the adjacency matrix, i.e., $ 1, 2, \dots, {\mj}$, and all newly added nodes labelled as $ \mj+1, \mj+2, \dots, \mj + \mc$.
    \item For all $e=(i, j) \in \mathcal{E}^+$, the newly added path from $i$ to $j$ is $p_{ij} = (i, \mj+k+1, \dots, \mj+k+l_{x,e}, j)$, where $k = \sum_{i = 1}^{e-1} l_{x,i}$.
    \item If the path has fixed direction such that $(j,i) \notin \mathcal{E}^+$, we define
    \begin{equation} \label{eq:directed_edge_oversample}
        \tilde{D}_{p_{ij}, p_{ij}} = \begin{bmatrix}
            0 & 1 & 0 & \cdots & 0\\
            0 & 0 & 1 & & 0\\
            \vdots & \vdots &\ddots & \ddots& \vdots\\
             0 & 0 & \cdots & 0 & 1\\
            0 & 0 & \cdots & 0 & 0
        \end{bmatrix},
    \end{equation}
    otherwise, if $(j,i) \in \mathcal{E}^+$, we define
        \begin{equation} \label{eq:undirected_edge_oversample}
        \tilde{D}_{p_{ij}, p_{ij}} = \begin{bmatrix}
            0 & 1 & 0 & \cdots & 0\\
            1 & 0 & 1 & & 0\\
            \vdots & \ddots &\ddots & \ddots& \vdots\\
             0 & \vdots & \ddots & 0 & 1\\
            0 & 0 & \cdots & 1 & 0
        \end{bmatrix}.
    \end{equation}
\end{enumerate}
 The new graph for the case of the AROMA DHN is partially illustrated in Figure~\ref{fig:oversampling}.}


\begin{remark} \normalfont
    Note that $\Tilde{D}_{p_{ij},p_{ij}} \in \mathbb{R}^{p_{ij}\times p_{ij}}$ samples the rows and columns of $\Tilde{D}$ in order of the elements of $p_{ij}$, meaning that blocks \eqref{eq:directed_edge_oversample} and \eqref{eq:undirected_edge_oversample} do not show up in $\Tilde{D}$ as in the introduced form. Also, for any elements in $l_x$ equal to zero, we  retrieve $p_{ij} = (i,j)$ in step 2.
    \hfill$\diamond$
\end{remark}

\begin{remark} \normalfont
    \cmt{The extension of the graph does not change the independent flow distribution in the network presented in Section~\ref{sec:hydraulics}. Without loss of generality and for the ease of notation, we use $q$ to denote the extended mass flow variable. Nonetheless, considering the introduced graph augmentation, the precise definition of the flow variable in the new setting is $q = \big[q_i \otimes \mathbf{1}_{p_i}\big]_{i \in \mathcal{E}^+ }$.
    \hfill$\diamond$}
\end{remark}

\subsubsection{From graph to state-space}
Let the graph resulting from $\Tilde{D}$ be denoted by $\Tilde{\mathcal{G}} = (\Tilde{\mathcal{N}}, \tilde{\mathcal{E}})$ with $|\tilde{\mathcal{N}}| = \mj + \mc$ number of state variables, where $\mj$ and $\mc$ denote the number of junctions and finite volume cells, respectively. Furthermore, let $\tilde{E}$ be the incidence matrix for the graph $\tilde{\mathcal{G}}$.

In terms of the interconnection between nodes, our aim is to capture the rate of change in energy flowing between the nodes. To this end, we introduce an intermediate edge variable, $ \varphi_e(q_e,x_e) = q_e x_e$, which indicates precisely the mentioned rate. However, since temperatures are inherently node-specific quantities, the temperature of the root node associated with each edge $e \in \mathcal{E}^+$ is required. Consequently, for the entire network, the rate of energy change across all edges is defined by
\begin{equation} \label{eq:energies}
    \varphi = \frac{1}{2}Q\left(|\tilde{E}|-\tilde{E}\right)^\top x,
\end{equation}
where \cm{$|\tilde{E}|$ denotes the element-wise absolute value operator of $\tilde{E}$, i.e., $|\tilde{E}|= [|\tilde{E}_{ik}|]$}, $x$ is the vector of state variables defined as $x = [x_{i}]_{i\in\tilde{\mathcal{N}}}$, and $Q$ is a diagonal matrix defined as $Q = \mathrm{diag}(q_i)_{i\in\tilde{\mathcal{E}}}$. Subsequently, by multiplying $\varphi$ with the full incidence matrix, we obtain the energy balance on each node. More precisely, we have
\begin{equation} \label{eq:node_energies}
    \phi = \tilde{E}\varphi,
\end{equation}
where $\phi \in \mathbb{R}^{|\tilde{\mathcal{N}}|}$ represents the nodal rate of change in energy due to in-flows and out-flows. By including external effects from heat exchangers and environment, we can introduce the complete state-space description through the thermal dynamics of the network.
To this end, given flow vector $q$, we define matrix $A(q)$ as
\begin{equation} \label{eq:Amatrix}
    \begin{aligned}
        A(q) = \frac{1}{2}\tilde{E} Q \left(|\tilde{E}|-\tilde{E}\right)^\top-D_\alpha,
    \end{aligned}
\end{equation}
where $D_\alpha$ is the diagonal matrix of heat loss coefficients characterized as $D_\alpha= \mathrm{diag}(\alpha_i)_{i\in\tilde{\mathcal{N}}}$.
Also, let $\mathcal{W}$ denote the set of nodes corresponding to heat exchangers in the network, i.e., $\mathcal{W} \coloneqq \{i \in \tilde{\mathcal{N}} : w_i \neq 0\}$.
Furthermore, we define an $|\tilde{\mathcal{N}}| \times |\mathcal{W}|$ matrix, denoted by $B$, with the entry in the $i^{\text{th}}$ row and $j^{\text{th}}$ column given by
\begin{equation}
    B_{ij} = \begin{cases}
        \dfrac{1}{\rho c_p}, &\quad \text{if }  j = 1,\ldots, |\mathcal{W}|  \text{ and } i = i_j,\\
        0, &\quad \text{otherwise}, 
    \end{cases}
\end{equation}
where $i_1,i_2,\ldots,i_{|\mathcal{W}|}$ are the elements of $\mathcal{W}$ sorted in an increasing order, i.e.,
$i_1<i_2<\cdots<i_{|\mathcal{W}|}$.
Accordingly, the thermal dynamics of the network can be described by
\begin{equation} \tag{P$_\mathrm{CT}$}\label{eq:CTN}
\begin{aligned}
    V\dot{x} &= A(q)x + B w,\\
    z &= Cx,
\end{aligned}
\end{equation}
where $V = \mathrm{diag}(V_i)_{i\in\tilde{\mathcal{N}}}$ is the volume matrix, and  $w = [w_{i}]_{i\in\mathcal{W}}$ is the vector of
external inputs and disturbances. \cm{Note that $\mathcal{W}$ can be partitioned into disjoint subsets $\mathcal{W}_C$ and $\mathcal{W}_P$ to distinguish between consumer and producer interactions, respectively.}
In Example \ref{ex:thermal_dynamics}, we discuss how to obtain $A$ for a simple case.

\begin{remark} \normalfont
    \cmt{The volume matrix $V$ is singular since $V_i = 0$ for $i \in \mathcal{N}$. One might suggest using a Schur decomposition to eliminate the algebraic equations associated with the network junctions. However, this approach requires inverting a sub-matrix of $A(q)$, and though $A(q)$ is diagonal and invertible, it depends on the control variable $q$, introducing nonlinear terms that significantly complicate the optimization problem.}
    \hfill$\diamond$
\end{remark}

\begin{example} \cmt{\normalfont\label{ex:thermal_dynamics}
Consider the graph in Figure~\ref{fig:simple_graph_example} with incidence matrix 
\begin{equation}
E = \begin{bmatrix}
    -1& 1& -1\\1 & -1 & 0\\0 & 0& 1
\end{bmatrix}.
\end{equation}
The component $\frac{1}{2}(|E| - E)$ takes all the leaving edges, namely the negative terms of $E$, and makes them positive. Then, through \eqref{eq:energies}, we obtain the rate of energy change for all the edges 
\begin{equation}
\varphi = Q \begin{bmatrix}
    1 & 0 & 0\\0& 1& 0\\1&0&0
\end{bmatrix}x.
\end{equation}
Subsequently, using \eqref{eq:node_energies}, we obtain the nodal rate of change in energy as
\begin{align}
    \phi = E\varphi =  { \begin{bmatrix}
    - q_{12} - q_{13}  & q_{21}&  0\\ q_{12} & -q_{21} & 0\\ q_{13} &0 &0
\end{bmatrix}}x. 
\end{align}
Finally, we obtain $A$ by subtracting the heat loss matrix $D_\alpha$ from $\phi$ as in \eqref{eq:Amatrix}.

\begin{figure}[thbp]
    \centering
    \includegraphics[width=.3\linewidth]{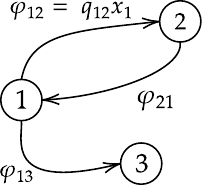}
    \caption{The energy transfer rate $\varphi$ in Example~\ref{ex:thermal_dynamics} illustrated as a function of edge flow rate and node temperature on a graph.}
    \label{fig:simple_graph_example}
\end{figure}}
\end{example}

\subsection{Time Discretization}
A discrete-time model is required for the MPC problem. Due to the slow thermal dynamics of DHNs, online computation constraints, and demand measurement intervals, practical implementations typically use time steps $\tau_t$ ranging from 15 minutes to 1 hour \cite{Krug2020}. For explicit discretization schemes, satisfying the Courant-Friedrichs-Lewy~(CFL) condition, i.e., $q_i(k){\tau_t} \leq V_i \; \forall i, \; k \in \mathbb{N}$, is difficult when modeling with large time steps. Hence, we employ the numerically stable implicit Euler method:
\begin{equation}
    x(k+1) = x(k) + \tau_t f_\mathrm{ct}(x(k+1),u(k+1)),
\end{equation}
where $f_\mathrm{ct}$ is the continuous-time dynamics. The application of the implicit Euler method to \eqref{eq:CTN} results in the following discrete-time system 
\begin{equation} \label{eq:DTS}
\begin{aligned}
    Vx(k+1) = Vx(k) +\tau_t \big[A(q(k))x(k+1) + Bw(k) \big],
\end{aligned}
\end{equation}
or, equivalently, we can split the algebraic part from the equation by introducing a subset of variables $x^{c}(k) \in \R^{\mc}$ to denote the state of thermal node cells with positive volume and {$\xj(k) \in \R^{\mj}$} to denote the state of junctions with zero volume. As a result, we can describe the system through a discrete-time differential algebraic equation~(DAE) as
\begin{equation} \tag{P$_\mathrm{DT}$} \label{eq:DT_dynamics}
\begin{aligned}
    \xc(k+1) &= f(\xc(k), \xj(k), u(k), d(k)),\\
    0 &= g(\xc(k),  \xj(k), u(k), d(k)).
\end{aligned}
\end{equation}
\cm{The system's dynamics are governed by the functions $f$ and $g$, which relate the state vector $x(k)= \big[
       \xc(k)^\top \  \xj(k)^\top
\big]^\top $, input vector $u(k) = \big[ q_r(k)^\top \  P(k)^\top \big]^\top$, and disturbance $d(k)$. A description of these variables is detailed in Table \ref{tab:system_variables}. }
Note that, in \eqref{eq:DT_dynamics}, the previously mentioned external input $w(k)$ has been split into a controllable power injection part $P(k)$ considered in $u(k)$ and a disturbance part considered in $d(k)$, which represents consumer demand. Additionally, $q_r(k)$ is the set of reduced flow variables from \eqref{eq:id_flows}.

\begin{table}[thpb]
    \centering
    \caption{Description of system variables.}
    \label{tab:system_variables}
    \vspace{-3mm}
    \ra{1.4}
    \begin{tabular}{cclcc}
    \toprule
        Symbol & \ \  & Description &  \ \ & Dimension\\
        \hline 
        $\xc(k)$&& State: Thermal node temperature   && $\mc$\\
        $\xj(k)$ && State: Junction temperature && $\mj$\\
        $q_r(k)$&& Control input: Volume flow rate && $\mr$\\
        $P(k)$ &&Control input: Power injection && $|\mathcal{W}_P|$\\
        $d(k)$&& Disturbance: Consumer demand  && $|\mathcal{W}_C|$\\
        $x(k)$&& Combined state vector  && $|\tilde{\mathcal{N}}|$\\
        $u(k)$ && Combined input vector && $\mr+|\mathcal{W}_P|$\\
        \bottomrule
    \end{tabular}
\end{table}

\section{Economic MPC Formulation for District Heating Networks} \label{sec:EMPC}
In this section, we introduce our economic MPC scheme for optimizing the performance of DHNs. To this end, we formulate the MPC optimization problem, discuss relevant design choices for the objective function, and present various computational techniques to improve the numerical performance of the employed solvers.

\subsection{Problem Formulation}
Consider the discrete-time dynamics 
\eqref{eq:DT_dynamics}. The receding horizon optimal control problem at time $k$ is defined as 
\begin{equation} \tag{OCP} \label{eq:OptimalControl}
    \begin{aligned}
        \min_{(x_t)_{t=1}^N, (u_t)_{t=0}^{N-1}} & \quad J_{N}\big((x_t)_{t=1}^N, (u_t)_{t=0}^{N-1}\big)\\
        \text{s.t.}\ \!
        & \quad  \xc_{t+1} = f(\xc_t, \xj_t, u_t, d_{t+k}),\\
        &\quad 0 = g(\xc_t, \xj_t, u_t, d_{t+k}),\\
        &\quad x_t \in \mathbb{X}_{t+k},   
        \\&
        \quad  u_t \in \mathbb{U}_{t+k},\\
        & \quad \forall t \in \{0,\dots, N-1\},\\
        & \quad x_0 = x(k),
    \end{aligned}
\end{equation}
\cm{where $x(k)$ is the current state at time $k$, the sets $\mathbb{X}_t \coloneqq \{x : \mathcal{F}_t(x) \leq 0\}$ are time-varying constraint sets representing operational bounds on the state, and the sets $\mathbb{U}_t \coloneqq \{u : \mathcal{H}_t(u) \leq 0\}$ represent time-varying operational constraints on the inputs, including the hydraulic constraints as defined in equations \eqref{eq:complementary_con}, \eqref{eq:id_flows}, \eqref{eq:matKhoff} from Section~\ref{sec:hydraulics}.}
Let the solution of the optimal control problem for horizon $N$ be denoted by $(x_{N,k}^*, u_{N,k}^*)$, where the subscripts are included 
to stress the dependence of optimal solutions on $N$ and $k$. 
In MPC, the optimal control problem in \eqref{eq:OptimalControl} is solved iteratively. In each step, the feedback control law is
\begin{equation*}
    \mu_{N}(x(k)) = u_{0}^*,
\end{equation*}
where $u_0^*$ is the first element of the optimal sequence $u_{N,k}^*$.

\begin{remark} \normalfont
    \cm{The optimal control problem yields an optimal volume flow rate sequence, $q(k)$. However, direct actuation of valves by the high-level MPC is impractical. Therefore, $q(k)$ serves as a setpoint for a lower-level valve controller.
     \hfill$\diamond$}
\end{remark}
%

\subsection{Objective Function} \label{ssec:objective_fun}
The objective function is a critical component in the design of economic MPC schemes, 
allowing for the selection of a cost function that accurately represents our practical goals or specifications.
Given the current objective of optimizing the performance of the DHN, we define the objective function as
\begin{equation}
    J_N = J_N^{\text{price}} + J_N^{\text{temp}} + J_N^{\text{diff}} + J_N^{\text{sto}}+ J_N^{\text{slack}},
\end{equation}
where each of these terms reflects a desirable operational feature or aspect, which are discussed below.
\begin{itemize}
    \item \textit{Price term}. Operational management requires minimization of operational costs. Additionally, we assume part of the generating mechanism is linked to a market, such as a heat pump purchasing from the electricity grid. Hence, we use a linear cost function
    \begin{equation} \label{eq:cost_price}
        J_N^{\text{price}} = \sum_{t=0}^{N-1} R_t^{\text{price}} \cm{P_t},
    \end{equation}
    where $R_t^{\text{price}}$ represents the time-varying price, or relative price, of generating $P_t$.
    \item \textit{Temperature term}. It is desired to operate DHNs at low temperatures to improve their efficiency. Therefore, we include a state term in the objective as
    \begin{equation}
        J_N^{\text{temp}} = \sum_{t=1}^{N} R_t^{\text{temp}} x_t^p 
    \end{equation}
    with the penalty coefficient $R_t^{\text{temp}}$, and power index $p \in \{1,2\}$, suggesting that the function can be linear or quadratic in $x$.
    \item \textit{Input variation term}. 
    It is undesirable to have fast switching in supply temperatures in DHNs, primarily due to the pipeline deterioration from the resulting thermal stress \cite{Zwan2020}. 
    Therefore, we consider a cost term as
    \begin{equation} \label{eq:input_var}
        J_N^{\text{diff}} = \sum_{t=0}^{N-2} R_t^{\text{diff}}(\cm{P_{t+1}- P_t})^2,
    \end{equation}
    which penalizes input deviations between any successive timesteps.
    \item \textit{Storage term}. \cmt{Minimization of operating costs and temperature usually does not favor charging of a storage buffer. Therefore, without any storage term, the MPC typically stays on discharging mode for all of the storage units. To address this issue, we introduce a terminal tracking cost on the temperature of storage nodes as
    \begin{equation}
        J_N^\text{sto} = \|x_N^{\text{sto}} - \overline{x}^{\text{sto}}\|_{R^\text{sto}}^2, 
    \end{equation}
    which is similar to the one employed in \cite{Labella2023}. This term will encourage the MPC to charge situationally depending on the size of $R^\text{sto}$.}
    \item \textit{Slack term}. The optimization is a large-scale nonlinear program, where the dynamics and constraints change in each iteration due to time-varying elements. Accordingly, guaranteeing feasibility in every iteration is not always possible.
    On the other hand, certain constraints, e.g., temperature bounds and demand satisfaction, may not be \emph{hard} constraints, meaning that a certain degree of violation is allowed. For this, we introduce variables $\sigma$ and a cost term as
    \begin{equation}
        J_N^\text{slack} = \|\sigma\|_{R^{\text{slack}}}^2,
    \end{equation}
    with $R^{\text{slack}}$ chosen to be large enough penalizing undesired constraint violations.
\end{itemize}


\subsection{Improving Numerical Performance}
The design of MPC algorithms, which inherently involve the iterative solution to an optimal control problem, depend on the convergence speed of the numerical solvers 
used to obtain the solution.
To this end, we implement various methods to improve the computational performance, as discussed below.

\begin{itemize}
    \item \textit{Warm starting}. This technique can improve computational speed by providing the numerical solver with a near-optimal initial guess. After a single iteration, we possess the $N$-step prediction of the state $x$. Thus, in the next iterate, we initialize the solver with 
    \mbox{$x_0^{k+1}, x_1^{k+1}, \dots x_{N-1}^{k+1} \leftarrow x_0^{k*}, x_1^{k*}, \dots x_{N-1}^{k*}$}, where $k$ denotes the index for current iteration. This process is repeated for all subsequent iterations.
    \item \textit{Objective terms}. Certain objective terms can assist the numerical solver in finding a solution. In particular, we noticed that the input variation term \eqref{eq:input_var} significantly improves the speed of convergence.
    \item \cmt{\textit{Control horizon}. Supply temperature changes can take a while to arrive at distant consumers. 
    Therefore, the prediction horizon should be sufficiently long to ensure that the effects of these changes can reach all consumers in the allocated horizon. 
    In order to limit the computational complexity, we can reduce the degrees of freedom in the optimization problem by setting a control horizon $N_c \leq N$ such that any control input after $N_c$ is equal to $u_{N_c}$, i.e, $u_{N_c + i} = u_{N_c}$ for $i = 1, \dots, N-N_c$.}
    \item \cmt{\textit{Move blocking}. Similarly to the control horizon, move blocking is a technique to reduce the dimension of the optimization variable. More precisely, it works by fixing the control variable over multiple time steps, which is not necessarily at the end of the horizon. For instance, here we will use 
    $u_t = u_{t+1} = u_{t+2} = u_{t+3}$, 
    for all $t$ that are multiples of 4 and less than $N-3$,
    which means that control inputs are fixed for four time steps in a row. }
\end{itemize}
\section{Numerical Experiments and Results} \label{sec:results}
In this section, to assess and verify the economic and computational performance of our proposed methods, we perform suitably designed numerical experiments and simulation studies. 
To this end,
we compare the proposed method to existing control strategies in the literature, including single-producer MPCs (SP-MPC) algorithms, which are based on, or similar to, optimization-based controllers used in \cite{Krug2020} and \cite{Labella2023}. 
Additionally, we compare with a rule-based control (RBC) scheme implementation, 
which is close to the widely adopted approach in practice for the control of DHNs.
Furthermore, we perform a numerical study to evaluate the computational tractability of the proposed methods. In particular, we assess the impacts of spatial oversampling and changing the prediction horizon on the computational load and performance of the algorithms. 
The district heating network considered in our numerical experiments is the AROMA network, introduced in Section~\ref{ssec:AROMA} and illustrated in Figure~\ref{fig:aroma_sims}.

For our numerical experiments, we employ a standard laptop with an \textsl{Intel i7-1185G7} processor to
run the simulations and Julia to create the models.
We use the mathematical programming package \textsl{JuMP.jl} \cite{Lubin2023} to build the optimization problems,  \textsl{Ipopt} \cite{Wachter2006} to solve the nonlinear optimization problems, and \textsl{DifferentialEquations.jl} \cite{Rackauckas2017} to simulate the DHN between iterations. 

\begin{figure}[thbp]
    \centering
    \includegraphics[width = 0.8\linewidth,trim={4cm 6.5cm 4cm 6.5cm}, clip]{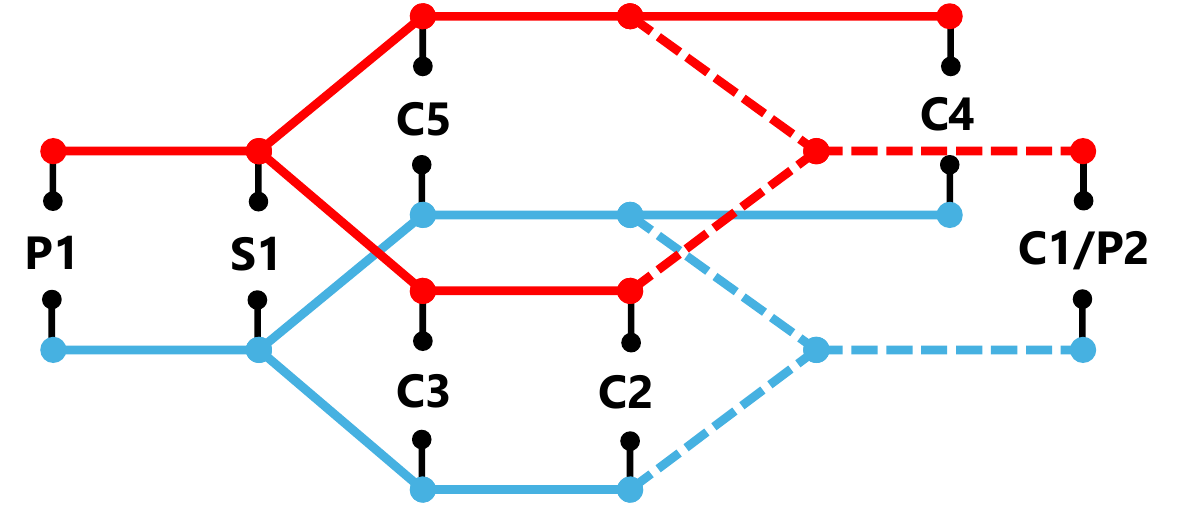}
    \caption{The AROMA network labeled with five consumers (labeled {C}), two producers (labeled P), and a storage (labeled S). Dashed lines indicate pipes on which bidirectional flows are allowed.}
    \label{fig:aroma_sims}
\end{figure}

Before proceeding with the results, we introduce the ground truth reference model to which we apply the generated inputs, and, secondly, we discuss the physical parameters, load profiles, electricity price profiles that were used in all simulations.

\subsection{High-fidelity Model}
To assess the performance of the controller, we will apply the inputs $\mu_N(x(k))$ to a high-resolution simulator that accurately describes the system. This simulator model, called here the \emph{high-fidelity model}, acts as a representation of the real system. In \cite{Simonsson2024}, a detailed study on the accuracy of these graph-theoretic simulation models for district heating systems is provided along with a comparison to the other high-fidelity simulators that have been verified using real measurements. It is shown in \cite{Simonsson2024} that, even for reduced-order models, their method exhibits high accuracy.

We refer to the high-fidelity model by $\mathrm{P_{CT}^\mathrm{HF}}$, which is a continuous-time system obtained using the same method as in \eqref{eq:CTN}. Nonetheless, the dimension of $x^{\mathrm{HF}}$, the state vector in the high-fidelity model, is equal to $\mj + \beta \mc$, where $\beta \in \mathbb{N}$, and $\tau_{x_i}^{\mathrm{HF}} = \tau_{x_i}/\beta$ to compensate for pipe length. This change suggests that for large values of $\beta$, we achieve a much higher spatial resolution, which leads to a better approximation of the original system. Finally, the resulting system $\mathrm{P_{CT}^\mathrm{HF}}$ is a system of differential algebraic equations that we solve using dedicated solvers in Julia \cite{Rackauckas2017}. In each iteration, the solver computes the evolution of the states for $\tau_t$ seconds. Every iteration is initialized using the final step of the previous simulation $x^\text{HF}(k\tau_t-\tau_t)$ and the optimal inputs $\mu_N(x(k))$ remain constant for the duration of the simulation, i.e., on the time interval $[k\tau_t - \tau_t, k\tau_t]$. After completion, the current state $x(k) = x^\text{HF}(k\tau_t)$ is fed into the MPC controller. See Figure~\ref{fig:block_scheme} for a diagram of this feedback loop.

\begin{figure}[thbp]
    \centering
    \includegraphics[width=.75\linewidth, trim={3cm 2.5cm 3cm 2.5cm}, clip]{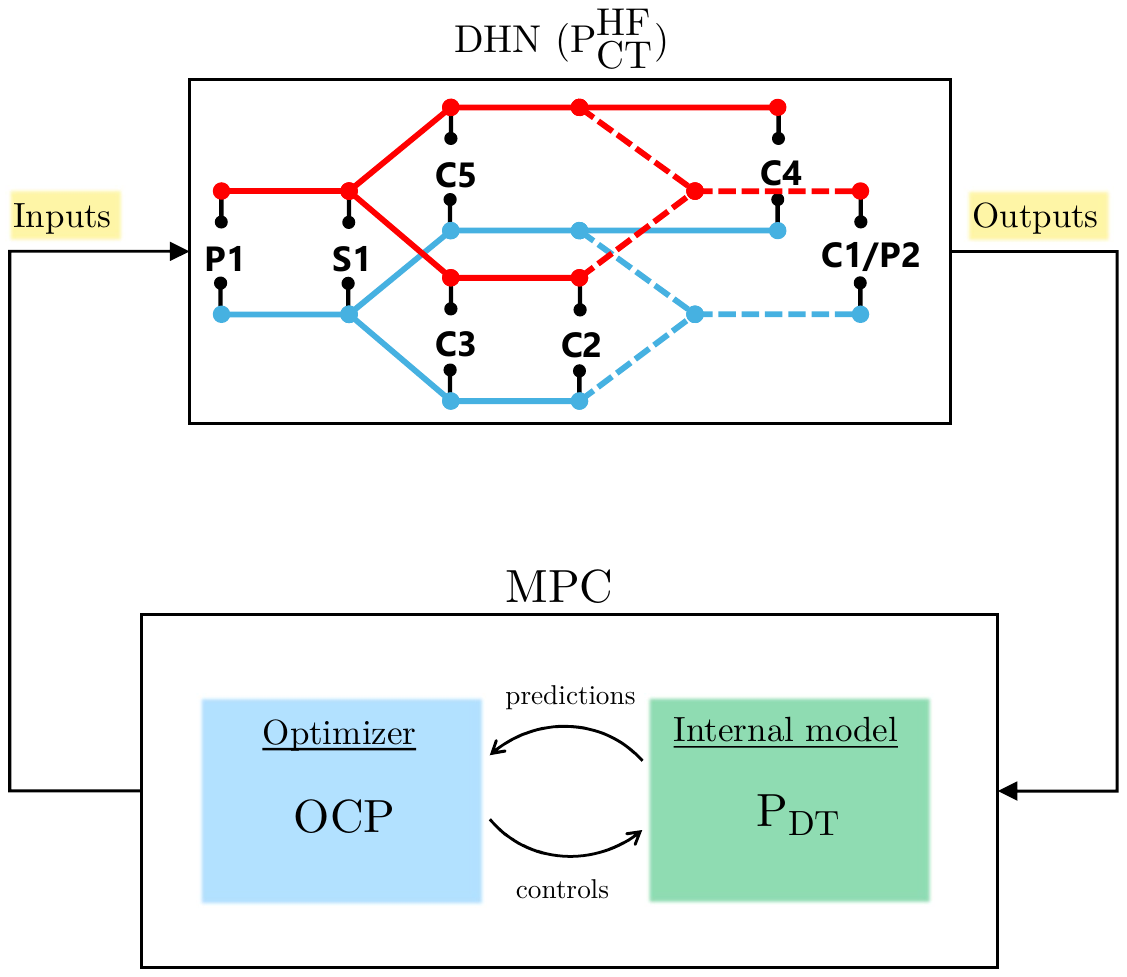}
    \caption{The block diagram showing the MPC feedback scheme.}
    \label{fig:block_scheme}
\end{figure}

\subsection{Parameters and Data}
In Table~\ref{tab:physical_parameters}, we list the physical parameters used in our simulations. System dimensions such as pipe diameters and lengths are the same as in \cite{Krug2020}. We obtain the heat transmission coefficient $U_{\text{pipe}}$ from \cite[p. 77]{Werner2013}, and the friction coefficient $K_\text{pipe}$ from \cite[p. 444]{Werner2013}. 
\begin{table}[thpb]
    \centering
    \caption{List of physical system parameters}
    \label{tab:physical_parameters}
    \vspace{-3mm}
    \ra{1.25}
    \begin{tabular}{lclcl}
    \toprule
        Parameter & \ \  & Value & \ \ & Units\\
        \hline 
        $T_a$&& 10 && $^\circ$C\\
        $\rho$ && 981 &&kg$\cdot$m$^\text{-3}$\\
        $c_p$&& 4182 &&J$\cdot$kg$^\text{-1} \cdot$K\\
        $U_\text{pipe}$ &&0.4 &&W$\cdot$m$^\text{-2} \cdot$ K$^\text{-1}$\\
        $K_\text{pipe}$&& 0.02 && -\\ 
        $d_\text{pipe}$&& 70 - 107 &&mm\\
        $d_\text{sto}$&& 2000 && mm\\
        $L_\text{pipe}$&& 300 - 600 &&m\\
        $L_\text{sto}$&& 8 && m\\
        \bottomrule
    \end{tabular}
\end{table}

The demand profile employed here is an approximation of the one used in \cite{Krug2020}. The electricity prices are acquired from Ember \cite{Ember2024}, where hourly electricity spot prices are provided for the Netherlands. \cm{We use price data from March 14, 2024 to determine the relative price term in \eqref{eq:cost_price}, i.e., $R_t^\text{price}$.}
Figure~\ref{fig:demand_price_plot} illustrates the demand and electricity prices for 24 hours, where only the net demand of consumers is shown. The demand of each individual consumer is computed as a fraction of the total demand as shown in Table~\ref{tab:dem_frac}, assuming the same load distribution as in \cite{Krug2020}.
\begin{table}[thpb]
    \ra{1.25}
    \caption{Fraction of total demand}
    \label{tab:dem_frac}
    \vspace{-2mm}
    \centering
    \begin{tabular}{c c c c c}
    \toprule
         C1/P2 & C2 & C3 & C4 & C5  \\
         \hline
         0.08&0.34&0.11&0.08&0.38 \\
    \bottomrule
    \end{tabular}
\end{table}

\begin{figure}[thpb]
    \centering
    \includegraphics[width=0.85\linewidth]{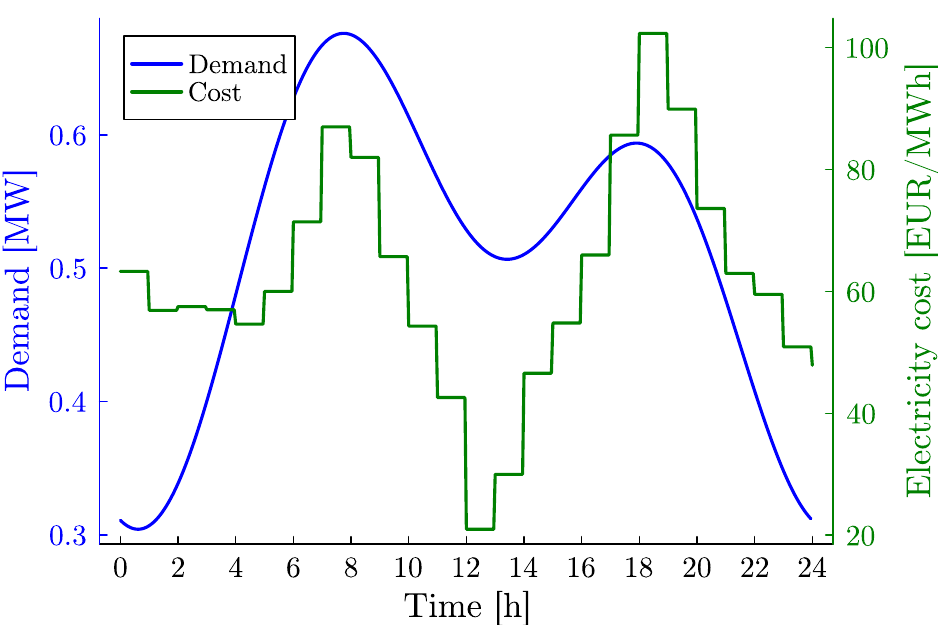}
    \caption{Total demand profile and electricity price profile for 24 hours. Electricity prices are from March 14, 2024, in the Netherlands.}
    \label{fig:demand_price_plot}
\end{figure}

\subsection{Economic Performance and Comparisons}

We conduct a comparative analysis of our proposed method against several established control strategies, including rule-based control and MPC for single producer DHNs without storage capabilities. The existing rule-based control strategy is referred to as RBC, while the single producer MPC strategy is denoted as SP-MPC. Our proposed methods are referred to as follows: the single-producer model with storage is denoted as SPS-MPC, the multi-producer model without storage is referred to as MP-MPC, and the multi-producer model with storage is indicated as MPS-MPC.

\cm{To quantify the economic value of the proposed methods, we analyze the financial implications for network operators under representative operating conditions. Here, prosumer C1/P2 generates a heat surplus of 100 kW during the interval 12:00-17:00, resulting in negative net demand. To maintain a similar total demand for comparative analysis, we introduce a compensatory 80 kW load increase at consumer C4.} 

\subsubsection{Cost comparison}
\cmt{In our numerical experiments, we set the control interval to $\Delta t = 15$ minutes. The cost function $J_N$ is formulated to prioritize economic performance, with the relative price term $R_t^\text{price}$ weighted substantially higher than other objective function terms, excluding slack variable penalties. A general overview of the economic performance results are presented in Figure~\ref{fig:bar_chart}.}

\begin{figure}[thpb]
    \centering
    \includegraphics[width = 0.87\linewidth]{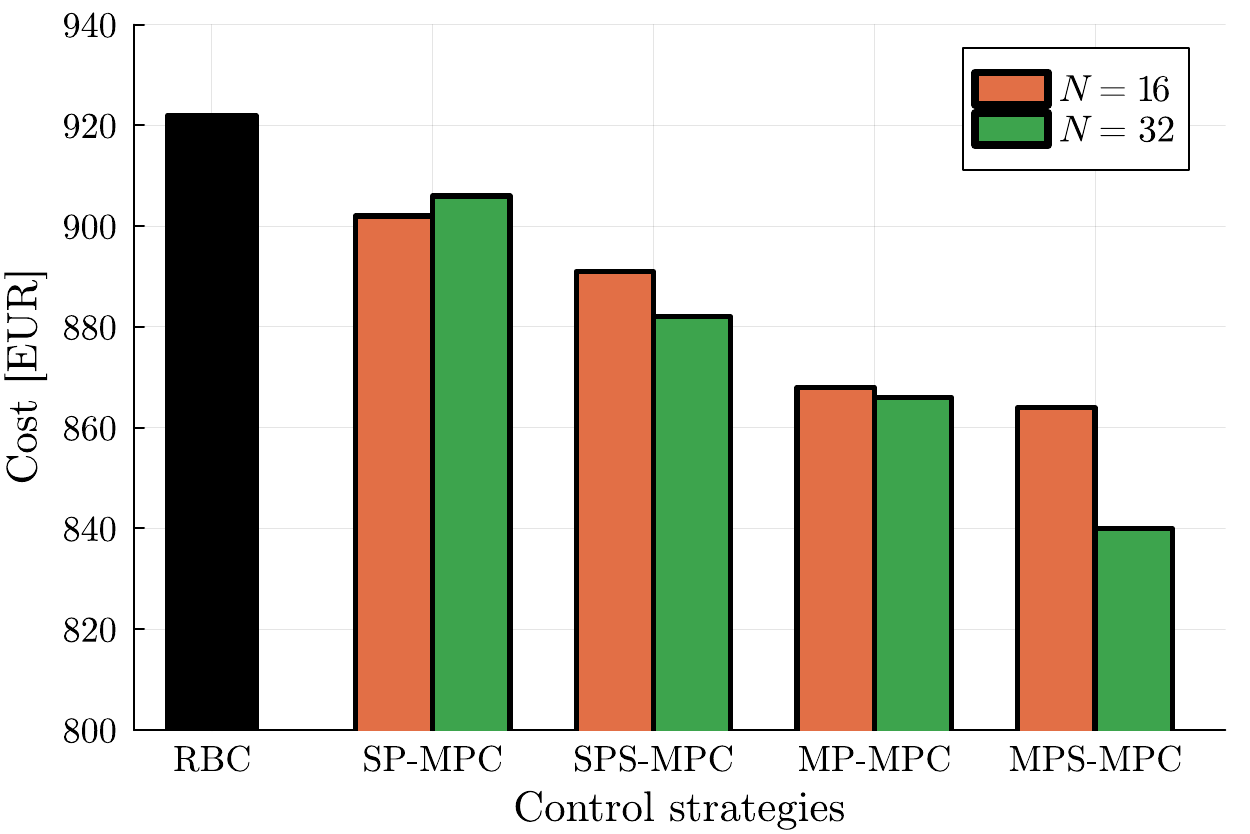}
    \caption{\cmt{Results of simulation 1. All results are based on 24 hours of simulation. Results from left to right are for: rule-based control (RBC), single-producer MPC (SP-MPC), single-producer MPC with storage (SPS-MPC), multi-producer MPC (MP-MPC), and multi-producer MPC with storage (MPS-MPC). The average runtime for each iteration of the MPC algorithms was under 10 seconds.}}
    \label{fig:bar_chart}
\end{figure}

\subsubsection{Quantifying the effect of slack variables}
\cm{To assess the true performance of the proposed methods, we quantify the extent to which the constraints have been violated as a consequence of the effects of the slack variables. 

In our analysis, we observe distinct patterns in constraint violations. Pumping capacity bottlenecks manifest primarily as demand violations, where the network cannot deliver sufficient flow rates to meet consumer requirements. Conversely, temperature violations - defined as failures to maintain minimum required supply temperatures at consumer substations - occur predominantly in scenarios where pumping capacity is sufficient to meet demand. While theoretical coupling between these constraint violations is possible, our simulations suggest they tend to be mutually exclusive, with temperature violations emerging only when demand is fully satisfied and demand violations occurring only under pumping capacity constraints.

Temperature requirements, typically mandated by regulatory frameworks, remain critical operational constraints that DHN operators must prioritize. For a simulation of $t_f$ time steps, we quantify these average temperature violations (ATV) in degrees Celsius as

\begin{equation}
        \text{ATV} = \frac{1}{t_f |\mathcal{W}_C|}\sum_{k = 1}^{t_f} \sum_{c \in \mathcal{W}_C} \text{max}(0, T_{\text{sup,min}} - T_{c}(k)). 
\end{equation}

On the other hand, network operators must ensure sufficient heat delivery to meet consumer demand requirements. We quantify demand violation (DV) over $t_f$ steps as

 \begin{equation*}
        \text{DV} = 100\% \times \frac{\sum_{k = 1}^{t_f} \sum_{j \in \mathcal{W}_C} d_j(k) - d_j^{\text{true}}(k)}{\sum_{k = 1}^{t_f} \sum_{j \in \mathcal{W}_C} d_j(k)},
    \end{equation*}
    where the true heat exchange is defined as
    \begin{equation*}
        d_j^{\text{true}}(k) = q_j(k) \left(T_{j-1}(k) - T_j(k)\right), \quad j \in \mathcal{W}_C.
    \end{equation*}
    In Figure \ref{fig:312}, we compare these violation metrics against the operational costs obtained in each method.
 }

\begin{figure*}[htbp]
    \captionsetup[subfigure]{labelfont={normalfont,small},textfont={normalfont,footnotesize} }
    \subfloat[\normalfont \cm{Daily operating costs versus temperature violations for a DHN with sufficient pumping capacity.}]{
        \includegraphics[width=0.41\textwidth]{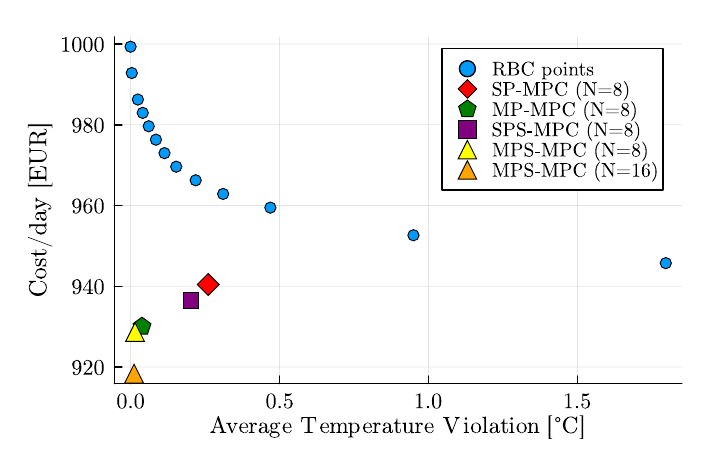}
        \label{fig:fig312a}
    }
    \hspace{0.1\textwidth}
    \subfloat[\normalfont \cm{Daily operating costs versus demand violations for a DHN operating with pumping capacity constraints.}]{
        \includegraphics[width=0.41\textwidth]{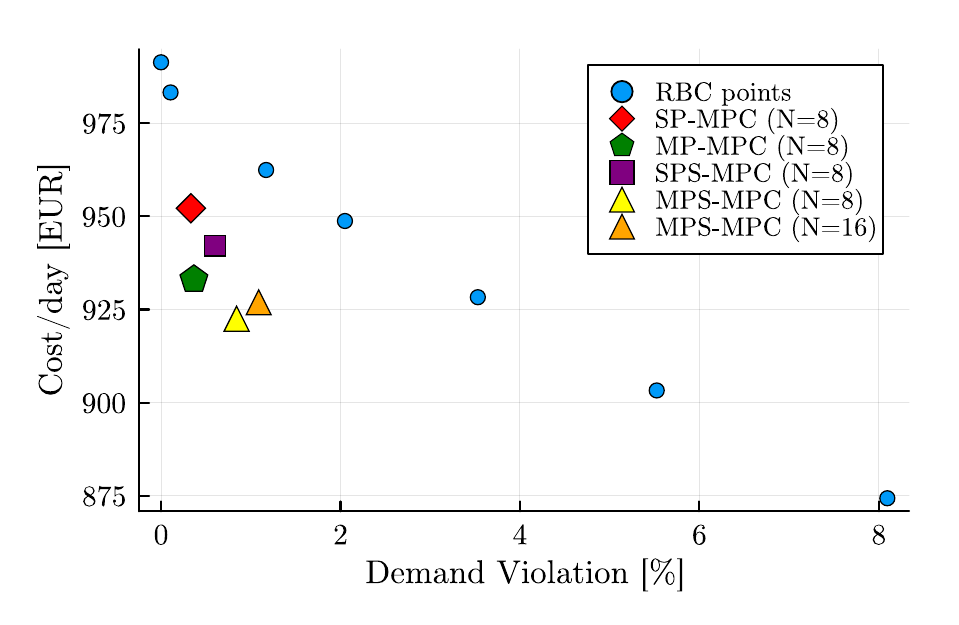}
        \label{fig:fig312b}
    }
    \caption{\cm{Relationship between operational costs and constraint violations under different pump configurations measured over a four day period.}}
    \label{fig:312}
\end{figure*}

\subsubsection*{Discussion} 

\cm{Several key observations emerge from Figure \ref{fig:bar_chart} and \ref{fig:312}. First, there is a noticeable improvement trend in cost reduction, with each feature added to the algorithm providing incremental benefits. The MPS-MPC algorithm with $N = 32$ achieves the highest performance, demonstrating an 9\% cost reduction compared to the rule-based controller.

Our analysis of constraint violations reveals that the MPC-based approaches actively optimize network operations based on predicted demand and price signals, resulting in dynamic temperature management rather than simple static setpoint tracking. This continuous optimization allows the controller to systematically reduce average network temperatures when beneficial while maintaining required service levels - a capability that rule-based controllers inherently lack. This dynamic temperature management is key to achieving both cost reductions and improved constraint satisfaction compared to rule-based approaches.

Additionally, when analyzing demand violations under pumping capacity constraints, the data suggest that prosumer-based control strategies, particularly MP-MPC and MPS-MPC, can better manage network limitations. The ability to actively coordinate multiple producers and redistribute flow patterns allows these approaches to excel at balancing competing objectives, demonstrating superior performance in maintaining service quality while optimizing operational costs.} 




\subsection{Added Value of Storage} \label{ssec:bens_storage}
The aim of this section is to demonstrate the benefits of storage in the DHN. As can be seen in Figure~\ref{fig:bar_chart}, the integration of storage yields significant economic benefits. In the top part of Figure~\ref{fig:benefits_of_storage}, the production schedules of the SP-MPC and SPS-MPC are plotted against each other. Secondly, the lower portion of the figure illustrates the corresponding storage charging and discharging periods. It is worth noting that, in this case, we have additionally implemented a constraint that ensures that the total charging volume approximately matches the total discharging volume over the course of the day. Further specific details regarding this constraint are provided in the subsequent discussion.

\begin{figure}[thbp]
    \centering
    \includegraphics[width = 0.8\linewidth, trim={0cm 0.6cm 0cm 0cm}, clip]{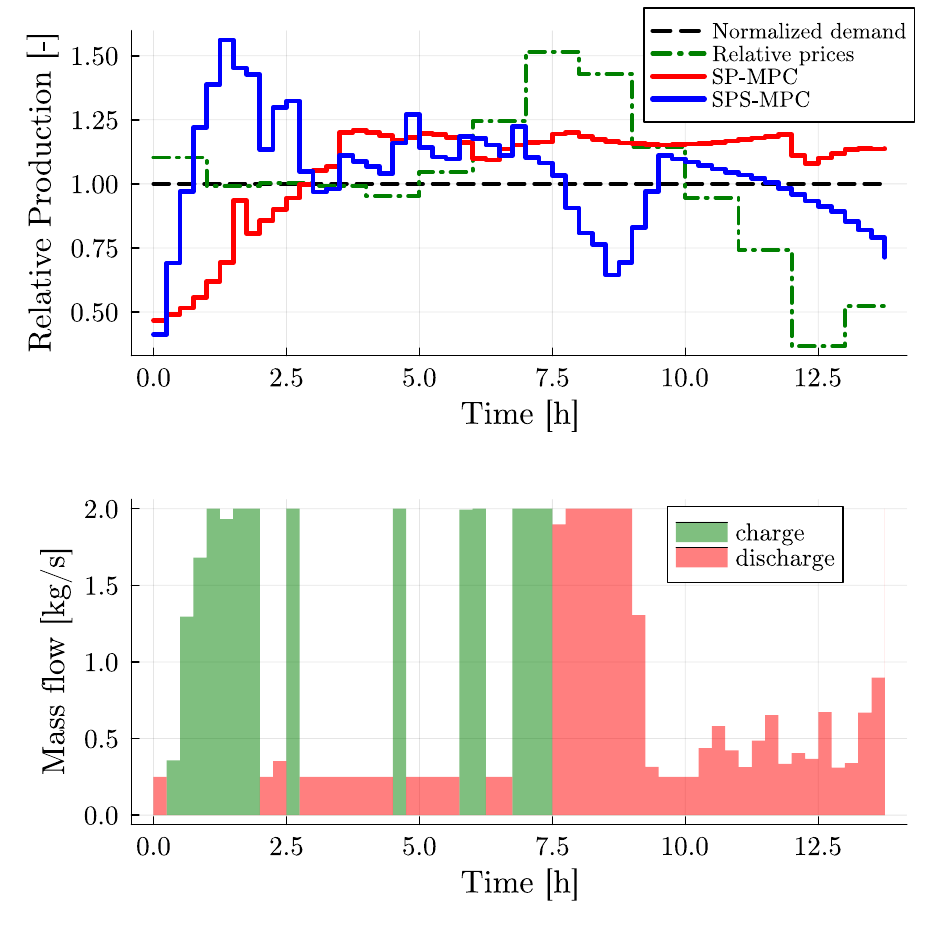}
    \caption{The top figure shows the power injected into the network relative to the demand being extracted for both the SP-MPC and the SPS-MPC methods. The green line denotes the normalized cost of generating heat and is shown for reference. The bottom figure shows the charging and discharging modes of the storage buffer in the SPS-MPC. For comparison the net mass flow entering or leaving the storage is equal to zero for the duration of the simulation.}
    \label{fig:benefits_of_storage}
\end{figure}
\subsubsection*{Discussion}
The results indicate that the controller leverages storage to enable flexible operational strategies. This flexibility is evidently seen in Figure~\ref{fig:benefits_of_storage}, where the comparison between the SPS-MPC (blue) and SP-MPC (red) in the top plot is presented. The SPS-MPC shifts part of the producer load by charging the storage at the beginning of the day. After approximately 7.5 hours, the controller discharges the storage, thereby reducing the load on the producer. Consequently, this allows the operator to achieve cost savings by minimizing production during periods of high prices.
Several limitations should also be acknowledged. This implementation requires additional constraints to ensure the total charging volume is equal to the discharging volume. Ideally, the storage term in the cost function, introduced in Section~\ref{ssec:objective_fun}, would suffice to incentivize the storage to maintain a sufficient level of hot water in the top layer of the storage. However, the effectiveness of this term is significantly influenced by the volume of the storage layers and the value of the objective weight. A weight that is too small results in constant discharging, whereas only an excessively large weight could occasionally induce charging behavior. Hence, in this study, we have opted to manually constrain the charging rates to emulate a more equal scenario.

\subsection{Added Value of Multiple Producers} \label{ssec:bens_multipro}


In addition to the benefit of distributing heat production, which allows facilities such as waste incineration plants or data centers to contribute to the district heating network, multiple producers can also reduce the pressure load on the central production plant by dividing the supply streams. The pressure drop per meter increases significantly with flow velocity (see, e.g., \cite[p. 442]{Werner2013}), so that each pipeline has a limited flow capacity. By allocating the heat supply flows among different producers, higher mass flows at consumer stations can be achieved, thus providing better operational margins. 

For this example, the setup has been slightly modified. Firstly, the prosumer C1/P2 now functions as a full producer, consistently generating a fixed amount of heat, i.e., 100 kW, for the DHN at all times. This scenario can be compared to a scenario where a waste incineration plant or a data center, typically located away from central power plants, is connected to the network and produces heat continuously throughout the day. Secondly, the available pumping power on each edge has been scaled down to emphasize the effects of the hydraulic constraints. Thirdly, the temperature bounds are tightened such that $70 \leq T \leq 90$ for all states in the system. Whenever these bounds are exceeded, it indicates that the optimization problem could not find a feasible solution within this range and had to surpass the maximum bound.

In Figure~\ref{fig:BD0_tempandflow}, the consumer temperatures and mass flows for the single-producer case are shown. On the other hand, Figure~\ref{fig:BD1_tempandflow} shows the temperatures and mass flows at consumers when C1/P2 contributes 100 kW constantly and P1 is freely controllable, corresponding to the multi-producer case.
\begin{figure}[thbp]
    \centering
    \includegraphics[width = 1\linewidth]{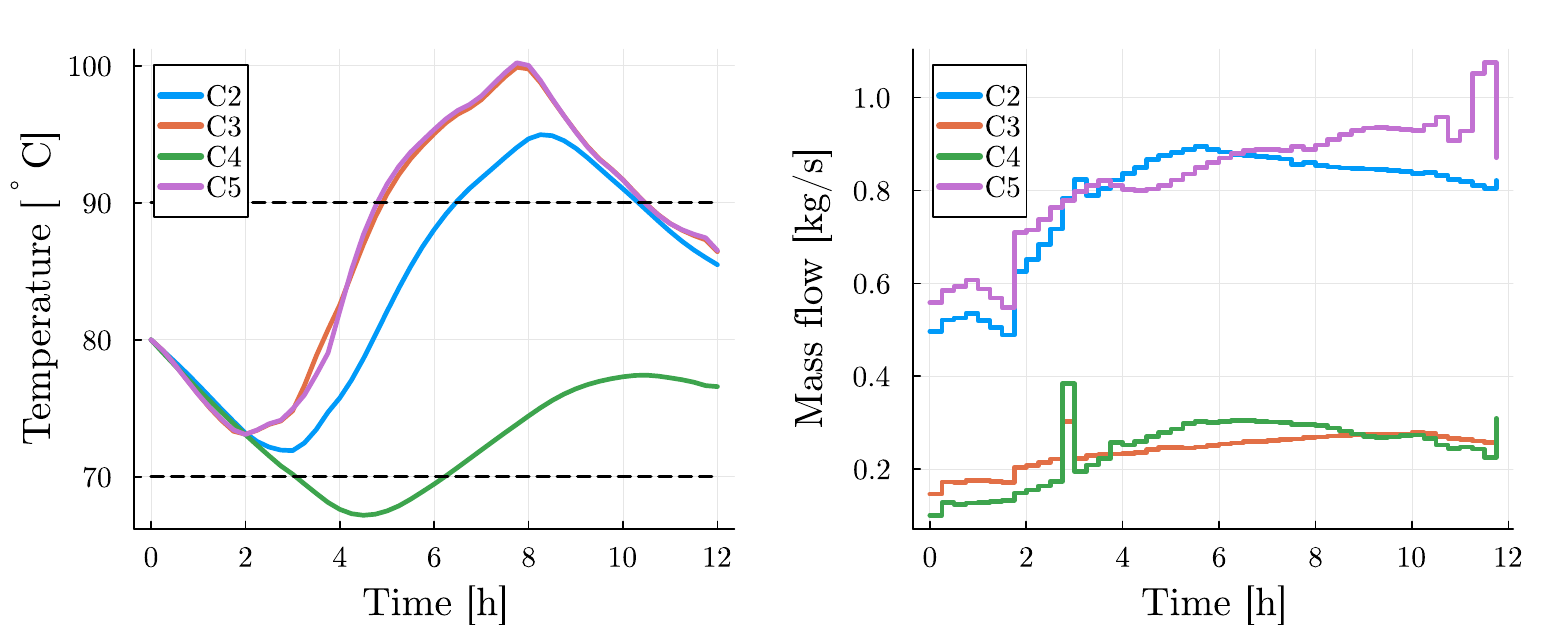}
    \caption{Consumer inlet temperatures and mass flows for a single-producer (SP-MPC) case.}
    \label{fig:BD0_tempandflow}
\end{figure}
\begin{figure}[thbp]
    \centering
    \includegraphics[width = 1\linewidth]{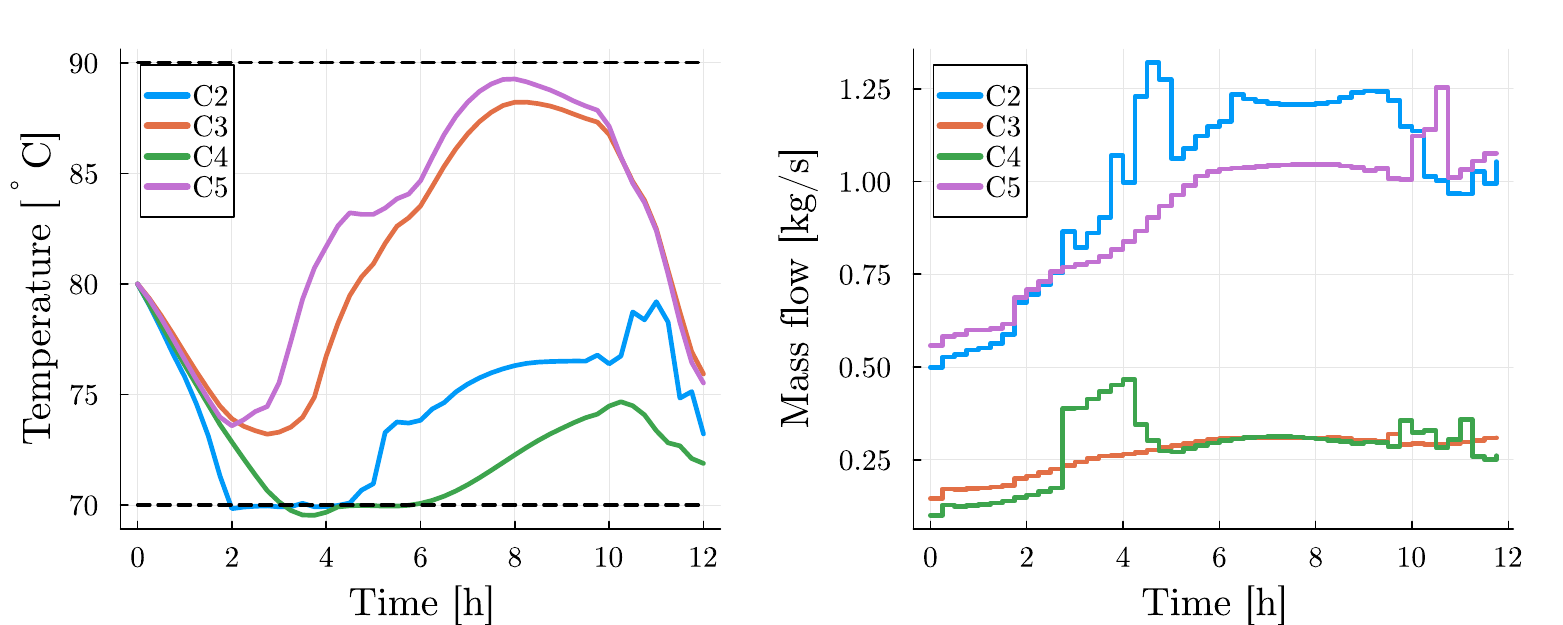}
    \caption{Consumer inlet temperatures and mass flows for a multi-producer (MP-MPC) case.}
    \label{fig:BD1_tempandflow}
\end{figure}

\subsubsection*{Discussion}
The results show that in the multi-producer scenario the network is able to deliver more heat to large consumers C2 and C5. In particular, the peak flow rate of C2 in Figure~\ref{fig:BD1_tempandflow} lies about 40\% higher than in Figure~\ref{fig:BD0_tempandflow}. The increase can be attributed to the fact that C2 receives heat from both P1 and C1/P2 sources. Consequently, P1's contribution to C2 is reduced, enabling it to allocate a greater share of its flow to C5, the largest consumer. Overall, the aggregate energy requirement in the multi-producer scenario, comprising the combined outputs of P1 and C1/P2, is roughly equivalent to the total production of P1 in the single-producer scenario. Therefore, the MP-MPC shows improved performance compared to the SP-MPC in staying within operational limits and fairly distributing heat.

\subsection{Computational Study}

Finally, we examine the computational performance and scalability of our algorithms for various model resolutions and prediction horizons. In Figure~\ref{fig:compstudy3}, the median computational times of the solver iterations for different model resolutions (left) and prediction horizons (right) are presented.

\begin{figure}[thbp]
    \centering
    \includegraphics[width = 1\linewidth]{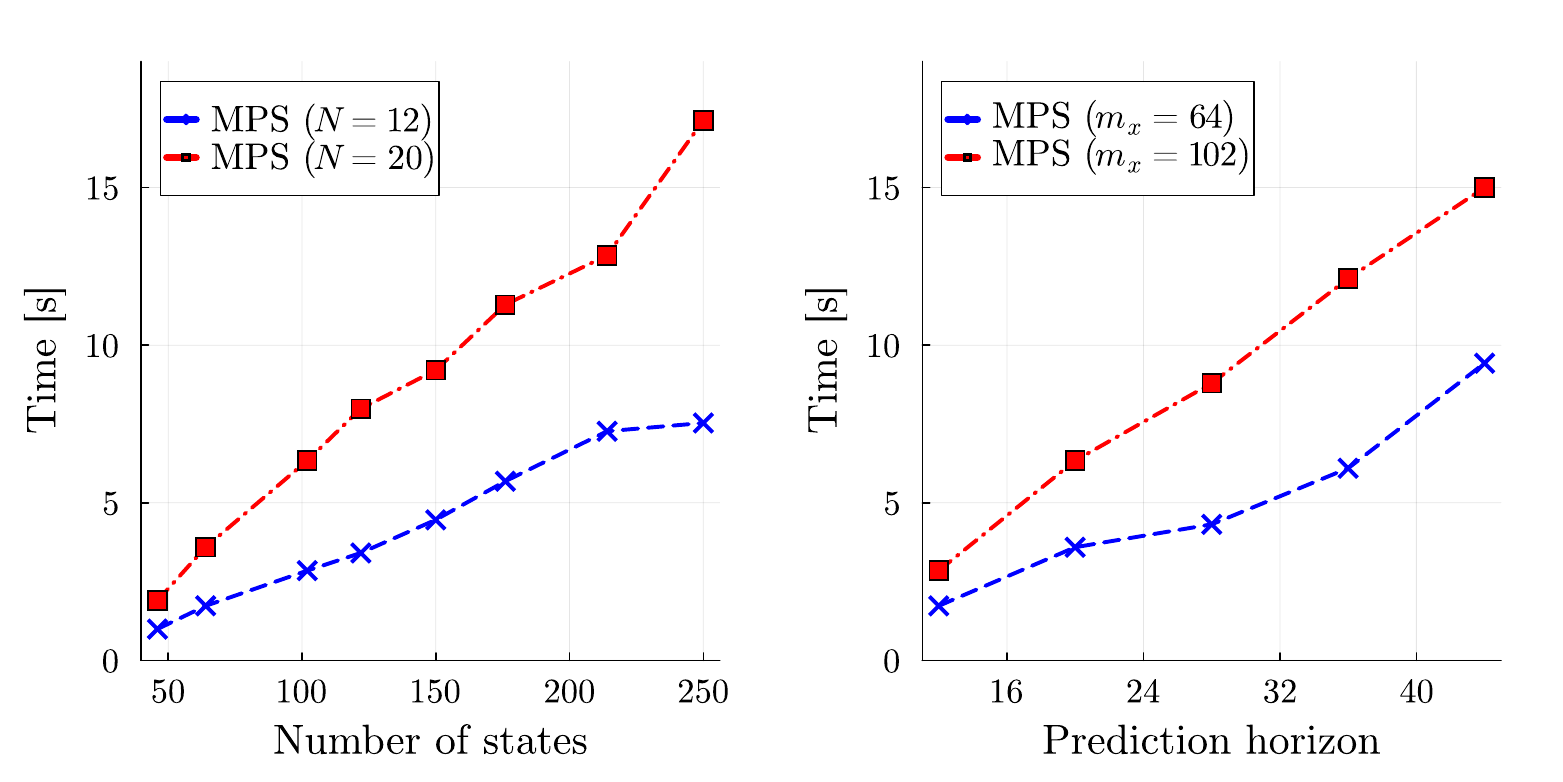}
    \caption{The left figure shows \textit{median solver} times for the MPS-MPC algorithm using a prediction horizon $N=12$ and $N=20$ for different number of state variables. The right figure shows for two models, one with 64 states and one with 102 states, the \textit{median solver} time based on the chosen prediction horizon.}
    \label{fig:compstudy3}
\end{figure}
\subsubsection*{Discussion}

In general, the computational cost associated with increasing model complexity remains manageable from an operational perspective, with even the most complex models requiring, on average, less than a minute per iteration. However, our observations indicate that the consistency of the solver in finding solutions within acceptable times diminishes under certain conditions, particularly for more complex models. These include models with over 200 states or with a prediction horizon greater than or equal to 40, where solution times can vary significantly. 

Assuming the existence of a feasible solution, a pragmatic approach to deal with varying solver times involves constraining the number of solver iterations or the solver time to ensure that controls are computed within the required timeframe, even if this may result in a decrease in solution quality. Nonetheless, increasing model complexity does not necessarily enhance performance. In certain scenarios, reducing the model resolution can actually decrease model mismatch. This is due to the fact that truncation error is influenced by the CFL condition, and this error is minimized when the CFL value is close to one \cite{Simonsson2024}. Consequently, the determination of model resolution is not straightforward and should be carefully considered on a case-by-case basis.

\subsection{Limitations and Future Work}

\cm{While our results demonstrate the effectiveness of MPC approaches for DHNs, several limitations should be acknowledged. The performance of the proposed methods relies on accurate demand predictions, with prediction errors affecting controller performance. Our implementation assumes both perfect knowledge of system parameters and access to full state measurements throughout the network. In practice, these assumptions may not hold, as many temperature and flow measurements might be unavailable or inaccurate.
From a computational perspective, solution times increase with model complexity and prediction horizon length, which may impact scalability for larger networks. Additionally, the non-convex nature of the optimization problem, particularly from bidirectional flow constraints, can cause the solver to struggle in finding optimal trajectories. This sometimes requires warm-starting strategies to guide the optimization toward desired flow configurations, an approach that introduces additional implementation complexity.

Future work could address these aspects through robust MPC formulations to handle demand uncertainty, state estimation techniques for networks with limited measurements, and development of computationally efficient optimization strategies.}

\section{Conclusion}
We developed an economic model predictive control algorithm designed for the operational management of district heating networks incorporating essential elements of 4th generation district heating networks such as multiple distributed heat sources, prosumers, and storage. A key aspect of our algorithm is its innovative treatment of hydraulic constraints through a convexification approach. Additionally, our method allows for the adjustment of model resolution to achieve the desired level of model accuracy. We have conducted comprehensive numerical experiments to evaluate the proposed features, \cm{demonstrating that MPC approaches significantly outperformed conventional rule-based controllers, yielding up to 9\% cost reduction alongside reduced constraint violations, with computation times remaining within practical limits. Furthermore, the integration of storage capabilities and multiple-producer configurations enhanced these performance metrics.}


\vfill

\end{document}